\documentclass[sigconf]{acmart}
\usepackage{cleveref}

\usepackage{bm}
\usepackage{bbm} 
\usepackage{caption}
\usepackage{subcaption}
\usepackage{booktabs,amsfonts,dcolumn}
\usepackage{graphicx}
\usepackage{subcaption}
\usepackage{mathtools}
\usepackage{wasysym}
\usepackage{hhline}

\usepackage{multirow}
\usepackage{booktabs,siunitx}

\usepackage{stmaryrd}
\usepackage{trimclip}
\usepackage{pifont}
\usepackage{enumitem}
\usepackage{mleftright}
\usepackage{scalerel,stackengine}
\usepackage{bbold}
\usepackage{bookmark}
\usepackage{algorithm}
\usepackage{algorithmic}

\renewcommand\footnotetextcopyrightpermission[1]{} 
\setcopyright{acmlicensed}
\copyrightyear{2018}
\acmYear{2018}
\acmDOI{XXXXXXX.XXXXXXX}
\acmConference[Conference acronym 'XX]{Make sure to enter the correct
  conference title from your rights confirmation email}{June 03--05,
  2018}{Woodstock, NY}
\acmISBN{978-1-4503-XXXX-X/18/06}

\begin{document}
\title{Directional Posterior Feedback Fusion for Query Reformulation in Dense Retrieval}

\author{Yiteng Tu}
\affiliation{%
    \institution{Tsinghua University}
    \city{Beijing}
    \country{China}
}
\email{tyt24@mails.tsinghua.edu.cn}

\author{Qingyi Pan}
\affiliation{%
    \institution{Tsinghua University}
    \city{Beijing}
    \country{China}
}

\author{Weihang Su}
\affiliation{%
    \institution{Tsinghua University}
    \city{Beijing}
    \country{China}
}

\author{Jiaxin Mao}
\affiliation{%
    \institution{Renmin University of China}
    \city{Beijing}
    \country{China}
}

\author{Yiqun Liu}
\affiliation{%
    \institution{Tsinghua University}
    \city{Beijing}
    \country{China}
}

\author{Qingyao Ai}
\authornote{Corresponding Author}
\affiliation{%
    \institution{Tsinghua University}
    \city{Beijing}
    \country{China}
}

\begin{abstract}
Dense retrieval has become a fundamental paradigm for modern information retrieval, supporting efficient semantic search across a wide range of applications.
Its effectiveness, however, relies critically on query embeddings that adequately capture users’ information needs, which can be difficult for short, ambiguous, or underspecified queries. 
Feedback-based query reformulation approaches like pseudo-relevance feedback (PRF) alleviate this by incorporating additional evidence from retrieved or generated content. 
However, existing embedding-space approaches commonly aggregate or interpolate feedback vectors, implicitly treating them as complete target representations. 
Such pointwise fusion can unnecessarily suppress query corrections along directions not expressed by individual feedback items, even when these directions are supported by complementary evidence. 
To address this limitation, we propose Directional Posterior Feedback Fusion (DIFF), a query reformulation framework that models feedback as directional evidence rather than target points. 
DIFF represents feedback as confidence-weighted directional evidence in the tangent space of the original query and jointly infers a controlled query update through prior-regularized posterior fusion, without retraining the underlying dense retriever.
Experiments across multiple dense retrievers and different feedback sources show that DIFF achieves the best overall retrieval performance, demonstrating the effectiveness of explicitly modeling how feedback should constrain query updates.
Code is available at:~\url{https://anonymous.4open.science/r/DIFF/}.
\end{abstract}


\maketitle

\section{Introduction}
Dense retrieval has become a fundamental paradigm for modern information retrieval~\cite{xiong2020approximate,izacard2021unsupervised,karpukhin2020dense}. 
It represents queries and documents as dense vectors and retrieves relevant documents according to their similarity in a shared semantic embedding space, thus bridging the gap between users' information needs and vast document collections~\cite{sun2026query,tu2026generalized}. 
However, its effectiveness depends critically on whether the query embedding adequately captures the user’s underlying information need~\cite{yu2021improving,lin2021contextualized}. 
User queries can be short, ambiguous, or underspecified, and only partially describe the relevant information~\cite{tu2026generalized,gao2023precise}. 
Query reformulation therefore offers a natural way to bridge this gap by enriching the initial query representation with additional evidence.

A common strategy is to obtain such evidence from feedback. 
Classical pseudo-relevance feedback (PRF)~\cite{abdul2004umass,rocchio1971relevance} assumes that top-ranked documents from an initial retrieval contain useful information and uses them to refine the query representation. 
More recently, large language models (LLMs)~\cite{yang2025qwen3,liu2024deepseek,achiam2023gpt,touvron2023llama} have further broadened the possible sources of feedback by generating hypothetical documents, pseudo-answers, query expansions, or other auxiliary content~\cite{gao2023precise,wang2023query2doc,mackie2023generative}.
Despite their different origins, these approaches share a common question: \textit{Given a set of feedback representations, how should they be used to modify the original query?}
In this work, we separate this question from how the feedback is obtained and focus on converting feedback vectors into a controlled and robust query update.

Existing vector-based feedback methods commonly construct the reformulated query by averaging feedback embeddings, linearly interpolating them with the original query, or otherwise moving the query toward an aggregate feedback representation~\cite{li2025llm,gao2023precise,yu2021improving}.
Although intuitive, such pointwise fusion implicitly treats each feedback vector as evidence for a complete target representation (as shown in the left part of Figure~\ref{fig:main}; more detailed discussion in \S\ref{subsec:old_obj}).
This assumption can be unnecessarily restrictive. 
A feedback item may support only one aspect of the information need rather than specify what the entire reformulated query should be. 
Consequently, dimensions or directions not expressed by that feedback should not automatically be interpreted as evidence against changing the query along those directions.

To address this limitation, we shift from pointwise fitting to directional constraints. 
Instead of asking the reformulated query to approach each feedback vector as a whole, we let each feedback constrain only the query movement along the direction it supports, as shown in the right part of Figure~\ref{fig:main}. 
Specifically, we operate in the tangent space of the unit sphere at the original query representation, where feedback can be expressed as local directional evidence and multiple pieces of evidence can be jointly incorporated through a constrained optimization objective. 
This formulation allows complementary feedback to govern different directions of the query update, while a prior toward the original query keeps the reformulation conservative when feedback evidence is insufficient.

\begin{figure}[t]
    \centering
    \includegraphics[width=0.5\textwidth]{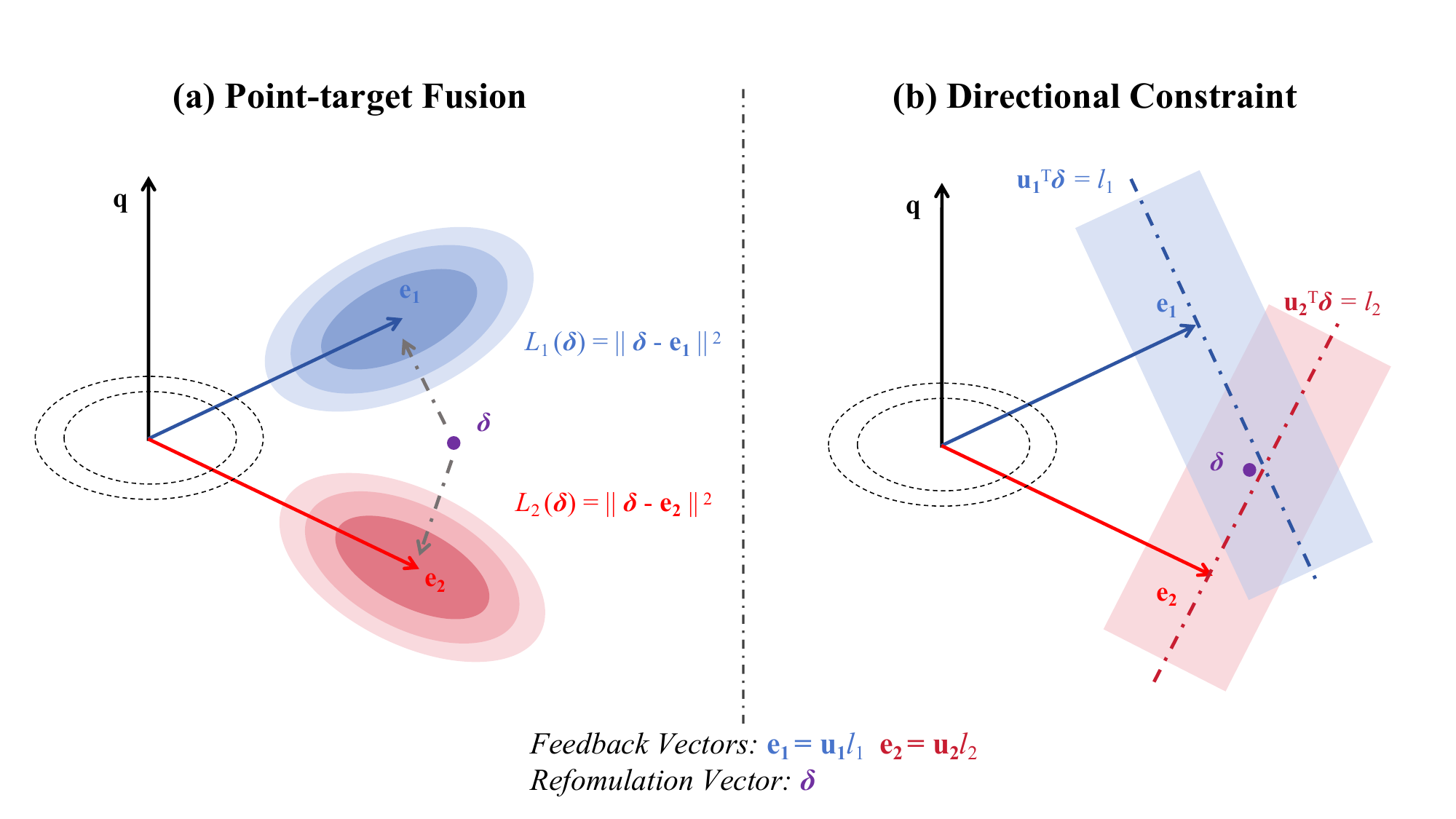}
    \vspace{-6mm}
    \caption{A comparison between pointwise fitting and directional constraints.}
    \label{fig:main}
    \vspace{-4mm}
\end{figure}

Based on this formulation, we propose \textbf{DI}rectional Posterior \textbf{F}eedback \textbf{F}usion (DIFF), a query reformulation framework for dense retrieval. 
DIFF first converts feedback vectors into directional evidence in the tangent space of the original query, and then estimates their contributions using lightweight local- and set-level confidence control. 
The resulting directional observations are jointly fused through a prior-regularized posterior objective with an efficient closed-form solution, after which the inferred correction is mapped back to the unit sphere. 
DIFF operates directly on fixed pretrained embeddings without updating the retriever or training an additional reformulation model, making it applicable to feedback from different sources.
Experiments across multiple retrieval benchmarks and dense retrievers show that DIFF achieves the best overall retrieval performance in diverse retrieval settings.

In conclusion, our contributions are threefold:
(1) We introduce a directional perspective on vector feedback for dense query reformulation. 
Rather than treating feedback embeddings as complete target points, we formulate each feedback item as a local directional constraint, allowing complementary evidence to contribute along different directions of the query update.
(2) We propose DIFF, a posterior fusion framework that jointly incorporates directional evidence with local feedback quality and set-level confidence, while explicitly regularizing the update toward the original query.
(3) We develop a training-free and feedback-source-agnostic reformulation method that operates directly on pretrained embeddings and admits an efficient closed-form solution, enabling its application to both retrieval-based and generated feedback.

\section{Preliminary}
\subsection{Dense Retrieval}
Given a query $x_q$ and a document $x_d$, a dense retriever $E(\cdot)$ encodes them into L2-normalizated $D$-dimensional representations:
\begin{equation}
    \mathbf q = E(x_q), \qquad \mathbf d = E(x_d),
\end{equation}
where $\mathbf q,\mathbf d\in\mathbb R^D$. 
Documents are then ranked according to the similarity between their representations and the query representation. With L2-normalization, $\|\mathbf q\|_2=\|\mathbf d\|_2=1$, so retrieval can be performed via their inner product, $s(\mathbf q,\mathbf d)=\mathbf q^\top\mathbf d$, which equals cosine similarity.
Under this normalization, embeddings lie on the unit sphere $\mathbb S^{D-1}$, and changing the query representation primarily amounts to changing its direction. 
This geometric view motivates us to model query reformulation locally around the original query rather than directly manipulating its vector magnitude.

\subsection{Feedback-based Query Reformulation}
Given an original query representation $\mathbf q$ and $K$ feedback representations generated by the same retriever $E(\cdot)$
\begin{equation}
    \mathcal F=\{\mathbf f_i\}_{i=1}^{K}, \qquad \mathbf f_i\in\mathbb R^D,\quad \|\mathbf f_i\|_2=1,
\end{equation}
A feedback-based query reformulator constructs a new normalized representation $\mathbf q'$, which is subsequently used with the same dense retriever to retrieve documents.

The feedback can originate from different sources. 
For example, the classical pseudo-relevance feedback (PRF) uses documents returned by an initial retrieval, whereas generative approaches can construct auxiliary texts such as hypothetical documents or pseudo answers and encode them as feedback. 
Our focus is orthogonal to how these feedback items are obtained: given $\mathbf q$ and $\mathcal F$, we study how to utilize them to produce a controlled query update regardless of how feedback is obtained.

A common strategy is to combine the original query with individual or aggregated feedback representations through averaging or linear interpolation. 
Abstractly, these approaches move the query toward one or more feedback vectors in the embedding space. 
As discussed above, this pointwise view implicitly treats a feedback representation as specifying a complete target for the reformulated query.
We next formalize why this can be undesirable when different feedback items provide complementary evidence.

\subsection{From Point Targets to Directional Evidence} \label{subsec:old_obj}
Since the query lies on the unit sphere, we consider its local tangent space $T_{\mathbf q}=\left\{\boldsymbol\delta\in\mathbb R^D: \mathbf q^\top\boldsymbol\delta=0 \right\}$, where $\boldsymbol\delta$ represents a local directional correction to the original query. 
Components parallel to $\mathbf q$ do not introduce a new local direction on the sphere, whereas components in $T_{\mathbf q}$ describe how the query direction can change.
Consider a feedback item whose proposed tangent-space correction can be represented as $\mathbf e_i=\ell_i\mathbf u_i$,
where $\mathbf u_i\in T_{\mathbf q}, \|\mathbf u_i\|_2=1$ is a unit direction and $\ell_i\geq0$ denotes the magnitude of evidence along that direction. 
A pointwise formulation would encourage the overall correction $\boldsymbol\delta$ to approach the entire feedback correction $\mathbf e_i$, for example by minimizing $\left\|\boldsymbol\delta-\ell_i\mathbf u_i\right\|^2$.
Because $\|\mathbf u_i\|=1$, this objective can be decomposed as (detailed derivation in Appendix~\ref{app:old_obj}): 
\begin{equation}\label{eq:decomp}
\left\|\boldsymbol\delta-\ell_i\mathbf u_i\right\|^2
=
\underbrace{
\left(\mathbf u_i^\top\boldsymbol\delta-\ell_i\right)^2
}_{\text{error along the supported direction}}
+
\underbrace{
\left\|
(I-\mathbf u_i\mathbf u_i^\top)\boldsymbol\delta
\right\|^2
}_{\text{penalty on all other directions}}.
\end{equation}

The first term has the desired effect: it encourages the query correction to match the evidence along the direction supported by feedback $i$. 
The second term, however, additionally penalizes every component of $\boldsymbol\delta$ orthogonal to $\mathbf u_i$. 
Such components are not necessarily undesirable. 
They may correspond to complementary directions supported by other feedback items. 
Thus, treating a feedback item as a complete target imposes constraints not only on what the feedback supports, but also on directions about which it provides no evidence.

To better explain this phenomenon, we compare the two ways of fusing $K$ feedback items under the \emph{same} original-query prior $\lambda$ and confidence weights $\tau_i\ge 0$, so that the only difference between them is what each feedback item is allowed to constrain:
\begin{equation}\label{eq:two_obj}
\begin{aligned}
\mathcal L_{\mathrm{pt}}(\boldsymbol\delta)
&=\frac{\lambda}{2}\|\boldsymbol\delta\|^2
+\frac12\sum_{i=1}^{K}\tau_i\left\|\boldsymbol\delta-\ell_i\mathbf u_i\right\|^2 ,\\
\mathcal L_{\mathrm{dir}}(\boldsymbol\delta)
&=\frac{\lambda}{2}\|\boldsymbol\delta\|^2
+\frac12\sum_{i=1}^{K}\tau_i\big(\mathbf u_i^\top\boldsymbol\delta-\ell_i\big)^2 .
\end{aligned}
\end{equation}
Here $\mathcal L_{\mathrm{pt}}$ fits every feedback target as a whole, whereas $\mathcal L_{\mathrm{dir}}$ keeps only the projection constraint $\mathbf u_i^\top\boldsymbol\delta\approx\ell_i$ of each item.
Let $\boldsymbol\delta^*_{\mathrm{pt}}$ and $\boldsymbol\delta^*_{\mathrm{dir}}$ denote their minimizers over $T_{\mathbf q}$, and let $\mathbf b=\sum_{i=1}^{K}\tau_i\ell_i\mathbf u_i$ be the confidence-weighted aggregate evidence, we have:
\begin{theorem}[Pointwise fusion interferes with complementary evidence]
\label{thm:interference}
Let $\lambda>0$, $\tau_i\ge 0$, $\ell_i>0$,
$\mathbf u_i\in T_{\mathbf q}$ with $\|\mathbf u_i\|=1$,
and $\mathbf b\neq\mathbf 0$.
Then, the minimizers satisfy: 
\begin{equation}\label{eq:thm_norm}
\|\boldsymbol\delta^*_{\mathrm{dir}}\|
\ge \|\boldsymbol\delta^*_{\mathrm{pt}}\|,
\end{equation}
with equality if and only if all supported directions are collinear.
Moreover, if the supported directions are mutually orthogonal,
then for every $j$ with $\tau_j>0$,
\begin{equation}\label{eq:thm_interf}
\mathbf u_j^\top\boldsymbol\delta^*_{\mathrm{dir}}
=\frac{\tau_j}{\lambda+\tau_j}\,\ell_j ,
\qquad
\mathbf u_j^\top\boldsymbol\delta^*_{\mathrm{pt}}
=\frac{\tau_j}{\lambda+\tau_j+\sum_{i\neq j}\tau_i}\,\ell_j .
\end{equation}
\end{theorem}

Theorem~\ref{thm:interference} (proved in Appendix~\ref{app:proof}) shows that the point-target and directional formulations agree only when the feedback set carries no complementary information. 
As soon as feedback items support different directions, pointwise fusion shrinks each item’s evidence by the total confidence $\sum_{i\neq j}\tau_i$ of \emph{all other} items—with $K$ equally confident complementary items, the recovered fraction $\tau/(\lambda+K\tau)$ vanishes as $K$ grows—whereas directional fusion recovers each direction at a rate $\tau_j/(\lambda+\tau_j)$ determined solely by its own confidence against the prior. 
The gap stems from the \emph{scope} of each constraint.
By Eq.~\eqref{eq:decomp}, a pointwise term penalizes $\boldsymbol\delta$ over the entire tangent space rather than only along $\mathbf u_i$. These penalties add $\tfrac12\boldsymbol\delta^\top P\boldsymbol\delta$ with $P=\sum_i\tau_i(I-\mathbf u_i\mathbf u_i^\top)$ to $\mathcal L_{\mathrm{dir}}$, so that every direction is penalized by the total confidence of the items that do \emph{not} support it.
Consequently, $\boldsymbol\delta^*_{\mathrm{pt}}=\mathbf b/(\lambda+\sum_i\tau_i)$ is an isotropically shrunk copy of the aggregate evidence. Its direction is fixed to that of $\mathbf b$ and its scale is set by the \emph{total} confidence. $\boldsymbol\delta^*_{\mathrm{dir}}=(\lambda I+\sum_i\tau_i\mathbf u_i\mathbf u_i^\top)^{-1}\mathbf b$ shrinks each supported direction only against its own confidence and the prior.
Directional fusion therefore covers a broader range of the complementary evidence, and indeed $\|\boldsymbol\delta^*_{\mathrm{dir}}\|\ge\|\boldsymbol\delta^*_{\mathrm{pt}}\|$, with equality only in the collinear case as shown in Corollary~\ref{cor:norm}.
Motivated by this analysis, we propose \emph{Directional Posterior Feedback Fusion (DIFF)}, which represents each feedback item as a selective directional constraint $\mathbf u_i^\top\boldsymbol\delta\approx\ell_i$ rather than a complete target point. 
It further regularizes the correction toward the original query and modulates the directional evidence according to feedback reliability, thereby preserving complementary evidence while remaining conservative under insufficient or unreliable feedback.

\section{Methodology}
\subsection{Overview}
Given a normalized query embedding $\mathbf q\in\mathbb R^D$ and $K$ normalized feedback embeddings $\mathcal F=\{\mathbf f_i\}_{i=1}^{K}$, our goal is to construct a reformulated query $\mathbf q'$ that incorporates useful information from the feedback while avoiding excessive deviation from the original query. 
As discussed above, we do not treat feedback embeddings as complete target representations.
Instead, we convert them into directional observations in the tangent space $T_{\mathbf q}$ and jointly infer a query correction $\boldsymbol\delta$ (as shown in Figure~\ref{fig:main}).
Specifically, each feedback item is represented by a triplet $(\mathbf u_i,\ell_i,\tau_i)$, where $\mathbf u_i\in T_{\mathbf q}$ denotes the direction supported by the feedback, $\ell_i$ specifies the target projection along this direction, and $\tau_i\geq0$ represents the confidence of this directional observation. Each feedback item therefore provides a noisy constraint $\mathbf u_i^\top\boldsymbol\delta \approx \ell_i$. We jointly fuse these observations by the objective:
\begin{equation}\label{eq:obj}
\boldsymbol\delta^*=\arg\min_{\boldsymbol\delta\in T_{\mathbf q}}\frac{\lambda}{2}\|\boldsymbol\delta\|^2
+\frac{1}{2}\sum_{i=1}^{K}\tau_i\left(\mathbf u_i^\top\boldsymbol\delta-\ell_i\right)^2,
\end{equation}
where $\lambda>0$ controls the strength of the original-query prior. 
The first term favors a conservative correction around the original query, whereas the second term fits the directional evidence according to its estimated reliability. 
Theorem~\ref{thm:interference} in \S\ref{subsec:old_obj} makes this precise: summing the above decomposition over all feedback items shows that pointwise fusion equals directional fusion plus the penalty $\tfrac12\boldsymbol\delta^\top P\boldsymbol\delta$ with $P=\sum_i\tau_i(I-\mathbf u_i\mathbf u_i^\top)$, which penalizes every direction in proportion to the total confidence of the feedback items that do not support it and is therefore what suppresses complementary corrections.
Equivalently, it can be interpreted as maximum a posteriori (MAP) inference under a zero-centered Gaussian prior over the query correction and noisy directional observations.
Our proposed method, Directional Posterior Feedback Fusion (DIFF), instantiates this formulation in three steps: constructing directional evidence $(\mathbf u_i,\ell_i)$ (\S\ref{subsec:direction}), estimating its precision $\tau_i$ through local- and set-level confidence (\S\ref{subsec:conf}), and efficiently solving Eq.~\eqref{eq:obj} while mapping the inferred correction back to the unit sphere (\S\ref{subsec:update}). Appendix~\ref{app:algo} summarizes the overall algorithm.

\subsection{Directional Evidence Construction} \label{subsec:direction}
We first extract the information that each feedback embedding introduces relative to the original query.
For feedback $\mathbf f_i$, its similarity to the query is $s_i=\mathbf q^\top\mathbf f_i$.
We then remove its component parallel to the query:
\begin{equation}
    \mathbf r_i=\mathbf f_i-s_i\mathbf q, 
\end{equation}
since $\mathbf q^\top\mathbf r_i=\mathbf q^\top\mathbf f_i-s_i\mathbf q^\top\mathbf q=0$, $\mathbf r_i \in T_{\mathbf q}$ and captures the component of feedback $i$ that can induce a local change in query direction.

We then decompose this tangent component into its direction and magnitude:
\begin{equation}
    \mathbf u_i=\frac{\mathbf r_i}{\|\mathbf r_i\|}, \qquad \ell_i=\|\mathbf r_i\|,
\end{equation}

Here, $\mathbf u_i$ determines where the query should be corrected, while $\ell_i$ determines the target projection supported along this direction. 

Since both $\mathbf q$ and $\mathbf f_i$ are normalized, we have $\|\mathbf r_i\| = \sqrt{1-(\mathbf q^\top\mathbf f_i)^2} = \sin\angle(\mathbf q,\mathbf f_i)$. 
Thus, $\ell_i$ measures the orthogonal novelty introduced by the feedback relative to the original query. 
Note that similarity $s_i$ and orthogonal evidence $(\mathbf u_i,\ell_i)$ play different roles in DIFF. 
The latter specifies the proposed query correction, whereas $s_i$ indicates how closely the feedback agrees with the original query and is subsequently used to estimate its reliability. 
This separation allows a feedback item to introduce new semantic directions without equating novelty itself with reliability.

\subsection{Hierarchical Feedback Confidence} \label{subsec:conf}
The directional representation above determines what correction each feedback item proposes, but not how much influence it should have on the reformulated query. 
Actually, feedback is inherently noisy and imperfect~\cite{tu2025robust,tu2026generalized}: some feedback items may be less aligned with the original information need, while even individually plausible items may collectively form an ambiguous or inconsistent correction signal. 
We therefore consider feedback reliability from two complementary perspectives. 
At the local item-level, we assess the relative quality of each feedback item with respect to the original query, so that stronger feedback contributes more to posterior inference. 
At the set-level, we assess whether the feedback items collectively provide a sufficiently coherent and reliable signal to justify modifying the query. 
The two levels respectively control the relative contribution of individual observations and the overall strength of the feedback set.
Accordingly, we factor the precision or confidence of the $i$-th directional observation as
\begin{equation}
    \tau_i = g\,a_i,
\end{equation}
where $a_i$ denotes its local quality and $g$ denotes the confidence of the entire feedback set. 

\subsubsection{Local Quality} \hfill\newline
At each feedback's local level, we would like to identify which feedback items are more compatible with the original query. 
However, directly applying a fixed threshold to the query-feedback similarity $s_i$ can be brittle, since similarity distributions may vary substantially across queries and embedding models. 
Instead, we evaluate each feedback item relative to the feedback set of the same query using robust MAD normalization $z_i=\frac{s_i-m}{1.4826\,M}$, where $m=\operatorname{median}_{i}(s_i)$ and $M=\operatorname{median}_{i}|s_i-m|$ are the median and median absolute deviation (MAD) of query-feedback similarities, respectively.
The robust normalization also helps reduce sensitivity to outlying similarity values and avoids relying on a global similarity scale.
Then, the local quality is calculated as $a_i=\sigma\left(z_i\right)$, where $\sigma(\cdot)$ is the sigmoid function. 
Thus, local quality $a_i$ captures a query-local notion of quality, feedback with relatively stronger query affinity receives greater influence, while relatively weak feedback is downweighted.

\subsubsection{Set-level Confidence} \hfill\newline
Local quality alone is insufficient to determine whether the overall feedback set as a whole supports a meaningful query update. 
In particular, a set may contain relatively strong individual items while their proposed corrections point in different directions. 
Therefore, we characterize set-level confidence using two factors: whether the directional evidence forms a coherent common signal, and whether the feedback set has sufficient average local quality.
We compute the mean evidence and summarize its common signal and dispersion:
\begin{equation}
    \boldsymbol\mu=\frac{1}{K}\sum_{i=1}^{K}\ell_i\mathbf u_i, \
    S=\|\boldsymbol\mu\|^2, \ V=\frac{1}{K}\sum_{i=1}^{K}\|\ell_i\mathbf u_i -\boldsymbol\mu\|^2.
\end{equation}
Here, $S$ measures the strength of the shared correction signal, whereas $V$ measures disagreement among individual feedback directions. 
We define the geometric confidence as $g_{\mathrm{geo}}=\frac{S}{S+V/K}$.
Thus, a coherent feedback set produces a strong common direction relative to its dispersion and therefore receives a high $g_{\mathrm{geo}}$, while conflicting or widely dispersed evidence yields a lower value.

However, geometric agreement alone is not sufficient.
A set of weak feedback items can still be mutually consistent. 
We therefore additionally compute its average local quality $\bar a=\frac{1}{K}\sum_{i=1}^{K}a_i$, and the overall set-level confidence is calculated as $g=g_{\mathrm{geo}}\bar a$.
$g$ is high only when the feedback exhibits both a coherent correction signal and sufficient query affinity. 
Together with $a_i$, it provides a simple hierarchical mechanism for controlling how strongly the directional evidence enters posterior inference.

\subsection{Posterior Inference and Spherical Update} \label{subsec:update}
After computing the features for each feedback, we can solve the directional posterior objective in Eq.~\eqref{eq:obj}. Let:
\begin{equation}
    U=\begin{bmatrix}\mathbf u_1^\top\\ \vdots\\ \mathbf u_K^\top \end{bmatrix} \in\mathbb R^{K\times D}, \
    W=\operatorname{diag} \left( \sqrt{\tau_1},\ldots,\sqrt{\tau_K} \right), \ 
    \boldsymbol\ell=[\ell_1,\ldots,\ell_K]^\top, 
\end{equation}
the solution to Eq.~\eqref{eq:obj} can be written as (detailed calculation in Appendix~\ref{app:sol}):
\begin{equation}\label{eq:opt_result}
    \boldsymbol\delta^*=U^\top W\left(\lambda I_K+WUU^\top W\right)^{-1}W\boldsymbol\ell.
\end{equation}

Rather than solving a $D\times D$ system in the embedding space, it requires only a $K\times K$ linear system. 
Since the number of feedback items is typically much smaller than the embedding dimension ($K\ll D$), the resulting complexity is $\mathcal O (K^2D+K^3)$, making posterior fusion efficient for practical feedback sizes.
This posterior solution determines a tangent direction and its unconstrained magnitude. 
To prevent uncertain feedback from causing an excessive change, we additionally use the set-level confidence to control the maximum angular update $\phi=\min\left(\|\boldsymbol\delta^*\|, \ g\,\theta_{\mathrm{up}}\right)$, where $\theta_{\mathrm{up}}$ is the maximum permitted update angle. 
Finally, we map the tangent correction back onto the unit sphere:
\begin{equation}
    \mathbf q'= \cos\phi\,\mathbf q + \sin\phi\, \frac{\boldsymbol\delta^*} {\|\boldsymbol\delta^*\|}.
\end{equation}

When $\boldsymbol\delta^*=\mathbf 0$, we simply set $\mathbf q'=\mathbf q$. 
Overall, DIFF separates query reformulation into three complementary decisions: which directions are supported by feedback, how strongly each piece of evidence should be trusted, and how far the query should ultimately move. 
Directional constraints preserve the complementary information carried by different feedback items, while hierarchical confidence adapts their influence according to both individual quality and collective reliability. 
The posterior objective then reconciles these observations with the original-query prior, and the final spherical update ensures that the resulting reformulation remains a controlled directional change in the normalized embedding space. 
Together, these components provide a unified mechanism for converting potentially noisy feedback vectors into a reliable query reformulation.

\section{Experimental Setup}
\subsection{Datasets and Evaluation Metrics}
We evaluate on 14 publicly available subsets of BEIR~\cite{thakur2021beir}, a widely used retrieval benchmark that covers diverse retrieval tasks and domains.
They are: ArguAna, Climate-FEVER, DBPedia, FEVER, FiQA, HotpotQA, MSMARCO, NFCorpus, Natural Questions, Quora, SCIDOCS, SciFact, Touche2020, and TREC-COVID. 
This collection provides a heterogeneous testbed spanning fact verification, question answering, argument retrieval, scientific retrieval, and other retrieval scenarios.
Following standard evaluation practice, we report NDCG@10 as the primary retrieval metric and Recall@100 as a complementary measure of retrieval coverage. 
Results are reported for individual datasets as well as aggregated across datasets when comparing the overall behavior of different reformulation methods.

\subsection{Retriever and Feedback Construction}
To examine whether our query reformulation method generalizes across models with varying capabilities and different embedding spaces, we experiment with three widely used dense retrievers: bge-base-en-v1.5\footnote{\url{https://huggingface.co/BAAI/bge-base-en-v1.5}}~\cite{xiao2024c}, e5-base-v2\footnote{\url{https://huggingface.co/intfloat/e5-base-v2}}~\cite{wang2022text}, and Qwen3-Embedding-0.6B\footnote{\url{https://huggingface.co/Qwen/Qwen3-Embedding-0.6B}}~\cite{zhang2025qwen3} (denoted as BGE, E5, and Qwen for simplicity).
For each retriever, queries, documents, and feedback texts are all encoded using the same underlying embedding model. 
All query and feedback representations are L2-normalized before reformulation, consistent with the unit-sphere formulation.
The reformulated query is used directly against the original document index, and no retriever parameters are updated during query reformulation.

We consider two complementary sources of feedback: pseudo-relevance feedback (PRF) obtained from first-pass retrieval and generative relevance feedback (GRF) produced by an LLM (i.e., Qwen3.8-27B\footnote{\url{https://huggingface.co/Qwen/Qwen3.8-27B}}~\cite{yang2025qwen3}). 
This allows us to separate the mechanism used to obtain feedback from the mechanism used to incorporate it into the query representation.
In PRF, for each query, we first perform dense retrieval using the original query representation and take the top-$K$ retrieved documents as the feedback. 
Each feedback document is encoded by the same dense retriever and supplied to the reformulation methods. 
For GRF, we consider three forms of generated feedback: HyDE, Q2D, and CoT~\cite{gao2023precise, wang2023query2doc, mackie2023generative}. 
This provides a complementary evaluation of whether DIFF and other fusion methods can operate across different feedback sources without modifying or retraining the reformulation mechanism.
We vary the number of feedback items over $K\in\{1,3,5,10,20\}$ to examine the behavior of each method as increasingly more feedback evidence—and potentially more noise—is introduced.
Specifically, we adopt a relatively moderate value of $K = 5$ in the main results (\S\ref{subsec:main_result}), and then examine the effects of different values of $K$ under various models and settings in the analysis (\S\ref{subsec:analysis}). 
The remaining implementation details are provided in Appendix~\ref{app:implement}.

\begin{table*}[htbp]
\centering
\caption{The average NDCG@10 and Recall@100 across all BEIR subsets, using different retrievers, feedback sources, and fusion methods. The best method for each retriever under different feedback sources is marked in bold and underlined. "*" indicates significantly better than the original retrieval performance without feedback and "$\dagger$" denotes significantly worse than the best method, both at the $p < 0.05$ level using the two-tailed pairwise t-test. Appendix~\ref{app:subset_detail} presents the specific results for each subset.}

\label{tab:main_results}
\resizebox{0.95\textwidth}{!}{
\begin{tabular}{l|cc|cc|cc|cc}
\toprule
Feedback from
& \multicolumn{2}{c|}{PRF}
& \multicolumn{2}{c|}{HyDE}
& \multicolumn{2}{c|}{Q2D}
& \multicolumn{2}{c}{CoT} \\
\midrule
Metrics & NDCG@10 & Recall@100 & NDCG@10 & Recall@100 & NDCG@10 & Recall@100 & NDCG@10 & Recall@100 \\
\midrule

\multicolumn{9}{l}{\textit{BGE: 0.5492 / 0.6873}}\\
\midrule
VPRF 
&0.5045$^{\dagger}$&0.6734$^{\dagger}$&0.5410$^{\dagger}$&0.6866$^{\dagger}$&0.5564$^{*\dagger}$&0.6968$^{*\dagger}$&0.5307$^{\dagger}$&0.6852$^{\dagger}$\\
Concat-P
&0.5294$^{\dagger}$&0.6825$^{\dagger}$&0.5494$^{\dagger}$&0.6871$^{\dagger}$&0.5550$^{*\dagger}$&0.6962$^{*\dagger}$&0.5350$^{\dagger}$&0.6849$^{\dagger}$\\
Concat-L
&0.5248$^{\dagger}$&0.6618$^{\dagger}$&0.5302$^{\dagger}$&0.6783$^{\dagger}$&0.5438$^{\dagger}$&0.6870$^{\dagger}$&0.5180$^{\dagger}$&0.6726$^{\dagger}$\\
OR-PRF
&0.5392$^{\dagger}$&0.6883$^{\dagger}$&0.5606$^{*\dagger}$&0.6954$^{*\dagger}$&0.5632$^{*\dagger}$&0.7001$^{*\dagger}$&0.5515$^{\dagger}$&0.6941$^{*\dagger}$\\
DIME
&0.5404$^{\dagger}$&0.6852$^{\dagger}$&0.5626$^{*\dagger}$&0.6954$^{*\dagger}$&0.5637$^{*\dagger}$&0.6961$^{*\dagger}$&0.5546$^{*\dagger}$&0.6921$^{*\dagger}$\\
Eclipse
&0.5388$^{\dagger}$&0.6792$^{\dagger}$&0.5641$^{*\dagger}$&0.6957$^{*\dagger}$&0.5638$^{*\dagger}$&0.6967$^{*\dagger}$&0.5552$^{*\dagger}$&0.6928$^{*\dagger}$\\
AnchorQE
&\underline{0.5466}$^{\dagger}$&\underline{0.6934}$^{\dagger}$&\underline{0.5700}$^{*\dagger}$&\underline{0.7027}$^{*}$&\underline{0.5698}$^{*\dagger}$&\underline{0.7027}$^{*\dagger}$&\underline{0.5646}$^{*}$&\underline{0.7000}$^{*}$\\
DIFF (ours)
&\textbf{0.5522}$^{*}$&\textbf{0.6943}$^{*}$&\textbf{0.5717}$^{*}$&\textbf{0.7034}$^{*}$&\textbf{0.5723}$^{*}$&\textbf{0.7057}$^{*}$&\textbf{0.5650}$^{*}$&\textbf{0.7003}$^{*}$\\

\midrule
\multicolumn{9}{l}{\textit{E5: 0.5231 / 0.6729}}\\
\midrule

VPRF 
&0.4775$^{\dagger}$&0.6579$^{\dagger}$&0.5088$^{\dagger}$&0.6687$^{\dagger}$&0.5214$^{\dagger}$&0.6769$^{*\dagger}$&0.5012$^{\dagger}$&0.6705$^{\dagger}$\\
Concat-P
&0.5045$^{\dagger}$&0.6626$^{\dagger}$&0.5045$^{\dagger}$&0.6549$^{\dagger}$&0.5118$^{\dagger}$&0.6688$^{\dagger}$&0.4801$^{\dagger}$&0.6441$^{\dagger}$\\
Concat-L
&0.5003$^{\dagger}$&0.6431$^{\dagger}$&0.4828$^{\dagger}$&0.6419$^{\dagger}$&0.4958$^{\dagger}$&0.6566$^{\dagger}$&0.4604$^{\dagger}$&0.6307$^{\dagger}$\\
OR-PRF
&0.5103$^{\dagger}$&0.6667$^{\dagger}$&0.5189$^{\dagger}$&0.6649$^{\dagger}$&0.5212$^{\dagger}$&0.6744$^{\dagger}$&0.5014$^{\dagger}$&0.6569$^{\dagger}$\\
DIME
&0.5117$^{\dagger}$&0.6691$^{\dagger}$&0.5383$^{*\dagger}$&0.6824$^{*\dagger}$&0.5373$^{*\dagger}$&0.6798$^{*\dagger}$&0.5314$^{*\dagger}$&0.6793$^{*\dagger}$\\
Eclipse
&0.5087$^{\dagger}$&0.6615$^{\dagger}$&0.5385$^{*\dagger}$&0.6832$^{*\dagger}$&0.5368$^{*\dagger}$&0.6788$^{*\dagger}$&0.5318$^{*\dagger}$&0.6789$^{*\dagger}$\\
AnchorQE
&\underline{0.5225}&\underline{0.6768}$^{\dagger}$&\underline{0.5487}$^{*}$&\underline{0.6884}$^{*}$&\underline{0.5488}$^{*}$&\underline{0.6879}$^{*\dagger}$&\underline{0.5388}$^{*}$&\underline{0.6860}$^{*}$\\
DIFF (ours)
&\textbf{0.5252}&\textbf{0.6796}$^{*}$&\textbf{0.5489}$^{*}$&\textbf{0.6889}$^{*}$&\textbf{0.5501}$^{*}$&\textbf{0.6908}$^{*}$&\textbf{0.5389}$^{*}$&\textbf{0.6874}$^{*}$\\

\midrule
\multicolumn{9}{l}{\textit{Qwen: 0.5492 / 0.6838}}\\
\midrule

VPRF
&0.5066$^{\dagger}$&0.6700$^{\dagger}$&0.5607$^{*\dagger}$&0.6932$^{*\dagger}$&0.5642$^{*\dagger}$&0.6984$^{*\dagger}$&0.5607$^{*\dagger}$&0.6983$^{*\dagger}$\\
Concat-P
&0.5282$^{\dagger}$&0.6811$^{\dagger}$&0.5652$^{*\dagger}$&\underline{0.7007}$^{*}$&\underline{0.5734}$^{*}$&\textbf{0.7040}$^{*}$&0.5606$^{*\dagger}$&0.6988$^{*\dagger}$\\
Concat-L
&0.5196$^{\dagger}$&0.6625$^{\dagger}$&0.5490$^{\dagger}$&0.6925$^{*\dagger}$&0.5589$^{*\dagger}$&0.6978$^{*\dagger}$&0.5442$^{\dagger}$&0.6891$^{*\dagger}$\\
OR-PRF
&0.5388$^{\dagger}$&0.6861$^{\dagger}$&0.5702$^{*\dagger}$&0.6989$^{*}$&0.5699$^{*\dagger}$&0.6970$^{*\dagger}$&\underline{0.5714}$^{*}$&\textbf{0.7025}$^{*}$\\
DIME
&0.5393$^{\dagger}$&0.6810$^{\dagger}$&0.5692$^{*\dagger}$&0.6947$^{*\dagger}$&0.5670$^{*\dagger}$&0.6947$^{*\dagger}$&0.5649$^{*\dagger}$&0.6950$^{*\dagger}$\\
Eclipse
&0.5410$^{\dagger}$&0.6772$^{\dagger}$&0.5697$^{*\dagger}$&0.6936$^{*\dagger}$&0.5657$^{*\dagger}$&0.6933$^{*\dagger}$&0.5650$^{*\dagger}$&0.6952$^{*\dagger}$\\
AnchorQE
&\underline{0.5459}$^{\dagger}$&\underline{0.6865}&
\underline{0.5720}$^{*\dagger}$&0.6970$^{*\dagger}$&0.5687$^{*\dagger}$&0.6979$^{*\dagger}$&0.5663$^{*\dagger}$&0.6969$^{*\dagger}$\\
DIFF (ours)
&\textbf{0.5491}&\textbf{0.6882}$^{*}$&\textbf{0.5777}$^{*}$&\textbf{0.7012}$^{*}$&\textbf{0.5738}$^{*}$&\underline{0.7021}$^{*}$&\textbf{0.5727}$^{*}$&\underline{0.7014}$^{*}$\\
\bottomrule
\end{tabular}
}
\end{table*}

\subsection{Baselines}
We compare DIFF against the original dense retriever and a range of fusion-based training-free query reformulation methods. 
To ensure a controlled comparison, all methods use the same pretrained encoder, initial document ranking, feedback set, and document index under each experimental configuration. 
Specifically, these methods are:
\begin{itemize}[leftmargin=*]
\item \textbf{Original}. We use the original query embedding without feedback reformulation, providing the base retrieval performance of each dense retriever.
\item \textbf{VPRF}~\cite{li2023pseudo,li2025llm}. VPRF directly combines the original query representation with feedback document embeddings through vector-level averaging or interpolation.
\item \textbf{Concat-P}. Pairwise concatenation: concatenate the query with each feedback separately, feed the result to the dense retriever, and then average the embeddings of all query-document pairs.
\item \textbf{Concat-L}. Lairwise concatenation: concatenate the query with all feedback documents to obtain a global context, then feed it to the retriever (similar to ANCE-PRF~\cite{yu2021improving}).
\item \textbf{OR-PRF}~\cite{zhang2026query}. It constructs query updates from feedback components orthogonal to the original query, providing a relevant comparison for assessing whether DIFF benefits from jointly modeling directional evidence rather than directly aggregating orthogonal residuals.
\item \textbf{DIME}~\cite{faggioli2024dimension}. DIME selects the important individual embedding dimensions based on the product across each dimension of the query and feedback embeddings. 
\item \textbf{Eclipse}~\cite{d2025eclipse}. Based on DIME, Eclipse further incorporates a comparison between top positive feedback and bottom negative feedback when calculating dimension importance.
\item  \textbf{AnchorQE}~\cite{sun2026query}. AnchorQE estimates feedback contributions using query- or anchor-based retrieval scores to interpolate query and feedback embeddings. Note that AnchorQE requires additional retrieval operations to obtain feedback-based retrieval scores, introducing extra retrieval signals beyond the given feedback.
to other methods.
\end{itemize}
\begin{table}[htbp]
\centering
\caption{Ablation of DIFF components.}
\label{tab:ablation}
\resizebox{1.0\columnwidth}{!}{
\begin{tabular}{c|cc|cc}
\toprule
Feedback from & \multicolumn{2}{c|}{PRF} & \multicolumn{2}{c}{HyDE} \\
\midrule
Metrics & NDCG@10 & Recall@100 & NDCG@10 & Recall@100 \\
\midrule 
\multicolumn{5}{l}{\textit{BGE: 0.5492 / 0.6873}}\\
\midrule
DIFF & 0.5522 & 0.6943 & 0.5717 & 0.7034 \\
\midrule 
\begin{tabular}[c]{@{}c@{}} Point-target Fusion \\ $\boldsymbol\delta=\sum\tau_i \ell_i\mathbf u_i$\end{tabular}& 0.5439 & 0.6888 & 0.5646 & 0.6978 \\
\begin{tabular}[c]{@{}c@{}}w/o Confidence \\ $\tau_i=1$\end{tabular} & 0.5454 & 0.6925 & 0.5703 & 0.7013 \\
\begin{tabular}[c]{@{}c@{}}w/o Local Quality \\ $a_i=1$\end{tabular} & 0.5462 & 0.6932 & 0.5481 & 0.6926 \\
\begin{tabular}[c]{@{}c@{}}w/o Robust Norm\\ $a_i=\sigma(s_i)$\end{tabular} & 0.5487 & 0.6936 & 0.5684 & 0.7007 \\
\begin{tabular}[c]{@{}c@{}}w/o Set-level Confidence \\ $g=1$\end{tabular} & 0.5483 & 0.6933 & 0.5565 & 0.6953 \\
\begin{tabular}[c]{@{}c@{}}w/o Geometric Confidence \\ $g=\bar a$\end{tabular} & 0.5502 & 0.6925 & 0.5701 & 0.7012 \\
\begin{tabular}[c]{@{}c@{}}w/o Average Local Quality\\ $g=g_{\text{geo}}$\end{tabular} & 0.5495 & 0.6926 & 0.5580 & 0.6960 \\
\bottomrule
\end{tabular}
}
\end{table}

\section{Results and Analysis}
\subsection{Main Results} \label{subsec:main_result} 
Table~\ref{tab:main_results} presents the overall retrieval performance across all BEIR subsets under different dense retrievers and feedback sources. 
Overall, DIFF delivers the strongest and most consistent performance. 
This trend holds for both PRF and all three forms of GRF, as well as across different embedding models, suggesting that the proposed directional formulation is not tied to a particular retriever or feedback construction strategy.

These results also highlight the importance of how feedback is utilized. 
Generating high-quality feedback with LLMs alone is insufficient, and how to incorporate such feedback is equally important.
Specifically, simple vector averaging or concatenation does not consistently improve retrieval performance and can even degrade it, particularly when the feedback contains noisy information (especially PRF). 
Methods that more carefully control the contribution of feedback generally exhibit more robust behavior. 
DIFF further improves upon these approaches by treating feedback as directional evidence instead of complete target representations and by adapting its influence according to both individual feedback quality and set-level confidence.
In Appendix~\ref{app:change}, we further report the per-query changes in NDCG@10 ($\Delta$NDCG@10) relative to the original retrieval results, indicating that DIFF behaves conservatively under noisy feedback and allows more beneficial updates when the feedback is more reliable.
Its consistent behavior across PRF and GRF supports our motivation that a feedback item may provide useful evidence for a particular correction direction without constituting a complete replacement for the original query representation.

Among the competing methods, AnchorQE is the closest to DIFF in overall performance, but requires additional retrieval operations to estimate feedback contributions. 
These operations provide additional retrieval signals beyond the given feedback, particularly in the GRF setting where feedback is generated independently of the document collection. 
Moreover, these additional retrieval operations introduce substantial computational overhead, as further demonstrated in the efficiency analysis in Appendix~\ref{app:efficiency}.
In contrast, DIFF requires no additional retrieval operations and achieves stronger overall performance.
It suggests that effectively modeling feedback geometry and reliability can be more beneficial than incorporating additional retrieval signals for feedback weighting.

Taken together, these results demonstrate that DIFF provides a robust feedback-utilization mechanism across heterogeneous retrieval conditions. 
Rather than relying on a specific feedback source or additional retrieval signals, it consistently extracts useful query corrections from the feedback already available. 
This cross-retriever and cross-feedback-source behavior is particularly important for our goal of developing a general reformulation framework that can be applied on top of diverse dense retrieval systems.

\begin{figure*}[t]
    \centering
    \includegraphics[width=0.95\textwidth]{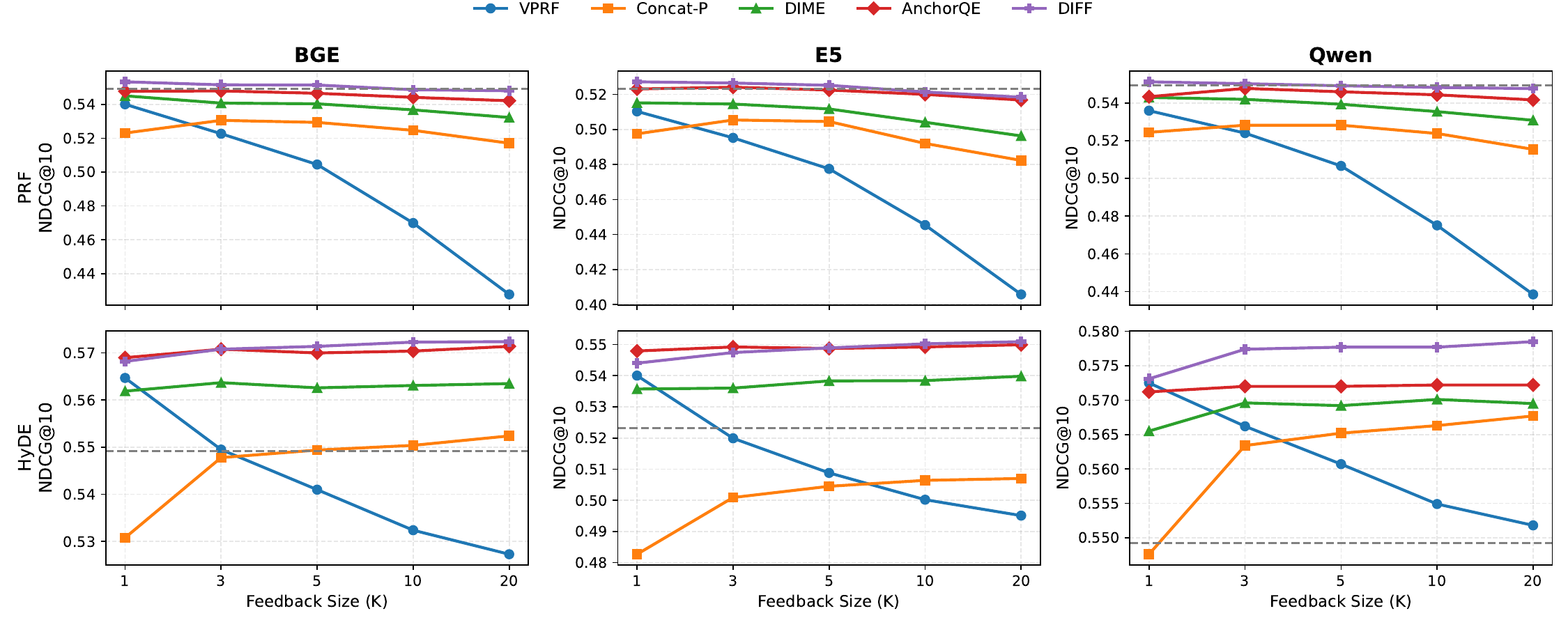}
    \caption{Effect of feedback size $K$ on the performace.}
    \label{fig:numfeed}
\end{figure*}

\subsection{Ablation Study}
We conduct an ablation study to examine the contribution of the major components. 
As shown in Table~\ref{tab:ablation}, the complete DIFF consistently achieves the best performance, while removing or replacing individual components generally degrades performance. 
It provides evidence for both key aspects of our design: representing feedback through directional constraints rather than point targets, and controlling its influence according to feedback reliability.

We first examine DIFF's central modeling choice: replacing directional posterior fusion with point-target fusion, where we directly aggregate the correction vectors as $\boldsymbol\delta=\sum_i\tau_i\ell_i\mathbf u_i$. 
This variant leads to one of the largest performance drops under both feedback sources.
Point-target fusion directly pulls the query toward feedback vectors and therefore implicitly suppresses components not expressed by each target. 
In contrast, DIFF imposes a more selective directional constraint, requiring the inferred correction to agree with each feedback item only along the direction it supports. 
The remaining degrees of freedom can consequently accommodate complementary evidence from other feedback items. 
The consistent advantage of DIFF over point-target fusion empirically supports this distinction between fitting feedback as a target and using feedback as directional evidence.

The remaining ablations demonstrate the importance of controlling feedback reliability at both the individual and set levels. 
Assigning equal precision to all feedback ($\tau_i=1$) causes more evident degradation under PRF than under HyDE, consistent with PRF containing noisier and more heterogeneous pseudo-relevant documents, whereas LLM-generated feedback is generally more consistent and focused on the input query. 
This suggests that confidence estimation becomes particularly valuable when feedback quality is uneven. 
Moreover, despite this degradation, the confidence-free variant still substantially outperforms simple point-target fusion methods such as VPRF and Concat, further demonstrating the advantage of directional fitting over directly aggregating feedback representations.
Within the hierarchy, removing local quality prevents DIFF from distinguishing the relative reliability of individual feedback items, while replacing robust normalization with raw query-feedback similarity also reduces performance, supporting query-wise calibration of feedback quality. 
Set-level confidence provides complementary control over whether the feedback collectively warrants a substantial update.
Removing either its geometric consistency or average local-quality component leads to degradation, and removing the set-level gate altogether is consistently inferior to the full model. 
Overall, these results show that the hierarchical design allows DIFF to account for both which feedback items should be trusted and how strongly the feedback set as a whole should influence the query, complementing the directional formulation with robustness to noisy or inconsistent feedback.

\subsection{Analysis} \label{subsec:analysis}
\subsubsection{Effect of Feedback Size ($K$)} \hfill \newline
We investigate the sensitivity of different reformulation methods to the number of feedback items $K\in\{1,3,5,10,20\}$. 
As shown in Figure~\ref{fig:numfeed}, DIFF remains consistently strong and stable across different feedback sizes and retrievers.
In contrast, several baselines are substantially more sensitive to $K$. 
This is particularly evident for VPRF under PRF, whose performance deteriorates rapidly as the feedback set grows. 
The different behaviors can be attributed to the quality of the underlying feedback: lower-ranked pseudo-relevant documents become increasingly noisy as more retrieval results are included, such that adding feedback does not necessarily provide additional useful evidence. 
Indeed, PRF-based reformulation frequently performs below the original retriever, and its performance generally decreases with larger $K$, further illustrating the risk of blindly incorporating pseudo-relevant documents. 
On the contrary, GRF exhibits a substantially different pattern: performance is generally higher and considerably more stable as $K$ increases, suggesting that the generated feedback provides more focused and reliable information than retrieved feedback. 
These observations also highlight that DIFF's directional formulation and hierarchical confidence mechanism allow additional feedback to be incorporated without forcing every feedback item to exert equal or unrestricted influence, making the method less susceptible to the increasing noise associated with larger feedback sets. 
Note that AnchorQE also exhibits stable performance with the aid of additional retrieval operations.
Overall, the results show that DIFF is effective for a particular choice of $K$, and it remains robust as the amount of feedback (together with noise) increases.
Additionally, an analysis of the other hyperparameters (i.e., $\lambda$ and $\theta_{\text{up}}$) can be found in Appendix~\ref{app:hyper}, and the efficiency comparison is reported in Appendix~\ref{app:efficiency}.

\begin{figure}[t]
    \centering
    \includegraphics[width=1.0\columnwidth]{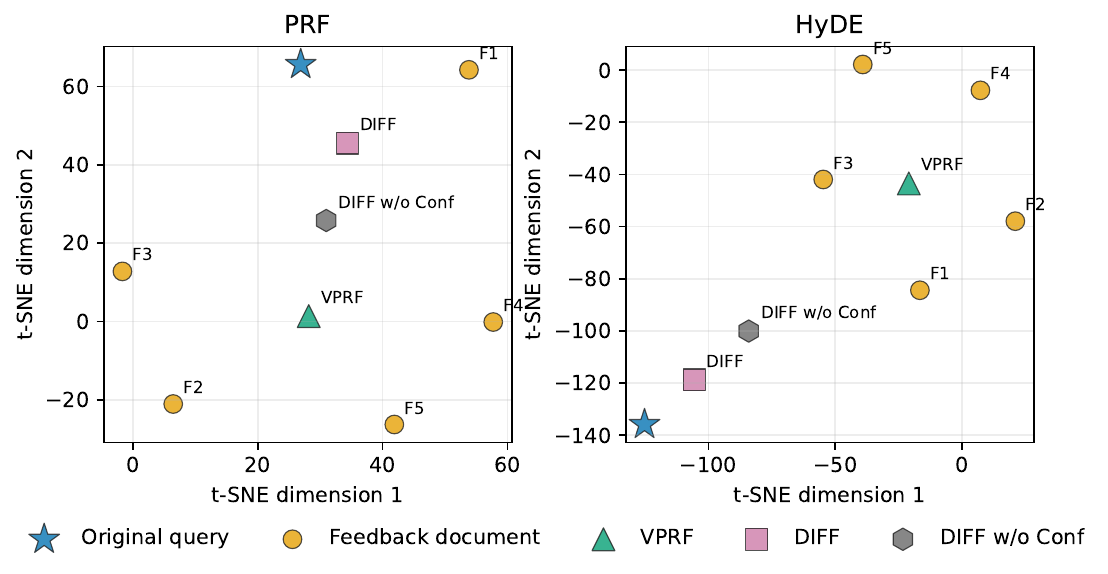}
    \vspace{-2mm}
    \caption{A case study of the geometric behavior.}
    \label{fig:tsne}
    \vspace{-2mm}
\end{figure}

\subsubsection{Case Study: Geometric Behavior of Query Reformulation} \hfill \newline
To provide an intuitive view of how different reformulation strategies modify the query representation, we visualize a representative case in Figure~\ref{fig:tsne}. 
We compare VPRF, DIFF without confidence control, and the full DIFF under both PRF and HyDE feedback. 
Notably, VPRF and DIFF w/o Conf differ primarily in their fusion mechanisms (i.e., Eq~\eqref{eq:two_obj} with $\tau_i=1$), allowing us to visually contrast the behaviors of these two strategies.
In PRF, the feedback documents are relatively dispersed, indicating substantial variation in the feedback signal. 
VPRF is pulled toward the feedback set as a whole, while DIFF produces a more controlled displacement from the original query. 
Removing confidence leads to a noticeably larger update, suggesting that directional fitting alone can preserve the structure of complementary feedback, whereas hierarchical confidence further limits the influence of uncertain evidence. 
In the HyDE case where generated feedback is more coherent, both DIFF variants move the query in a similar direction, but the full model still produces a more conservative update.
Overall, these cases illustrate how directional fusion produces selective query corrections, while confidence control adapts their magnitude to feedback reliability.
\section{Related Work}
\subsection{Dense Retrieval}
Dense retrieval represents queries and documents as continuous semantic vectors and performs retrieval according to their similarity in a shared embedding space~\cite{fang2024scaling,zhao2024dense}. 
Early dual-encoder approaches such as DPR demonstrated the effectiveness of independently encoding queries and passages for efficient maximum-inner-product search~\cite{karpukhin2020dense,zhan2020repbert}, while subsequent work improved dense retrievers through better negative sampling~\cite{xiong2020approximate,zhan2021optimizing,qu2021rocketqa} and large-scale pretraining~\cite{wang2022text,izacard2021unsupervised,xiao2022retromae,gao2021condenser,muennighoff2025generative}. 
More recent embedding models further improve generalization across heterogeneous retrieval tasks through large-scale weak supervision and instruction-aware representation learning~\cite{xiao2024c,zhang2025qwen3,behnamghader2024llm2vec}. 
Despite these advances, retrieval effectiveness still depends critically on the quality of the query representation: short or under-specified queries may not fully express the underlying information need. 
Our work is complementary to improvements in retriever training: we keep the pretrained dense retriever fixed and study how its query representation can be effectively refined using external feedback.

\subsection{Query Reformulation}
Query reformulation aims to enrich or modify the original query using additional information and can broadly be viewed through two sources of feedback~\cite{li2026query}. 
Pseudo-relevance feedback (PRF) assumes that top-ranked documents from an initial retrieval contain relevant information and uses them to refine the query. 
Classical approaches operate in the lexical space through term reweighting or expansion~\cite{abdul2004umass,rocchio1971relevance}, while recent methods extend PRF to neural and dense retrieval by learning feedback-aware query encoders or directly manipulating query and feedback representations in the embedding space~\cite{yu2021improving,li2025llm}. 
In parallel, generative relevance feedback (GRF) leverages the parametric knowledge of LLMs to generate auxiliary information without relying exclusively on the initial retrieval results. 
Representative approaches generate hypothetical documents, pseudo-documents or answers, query expansions, or other forms of relevant content that are subsequently incorporated into retrieval, such as HyDE~\cite{gao2023precise}, Query2Doc~\cite{wang2023query2doc}, CoT~\cite{mackie2023generative}, Lamer~\cite{shen2023large}, and GPRF~\cite{tu2026generalized}. 
These two lines of work have substantially expanded \textit{how useful feedback can be obtained}, from retrieved documents to LLM-generated knowledge.
However, once feedback is available, an equally important question is \textit{how to incorporate it to modify the query representation}, especially for dense retrievers that operate in the semantic embedding space beyond the explicit vocabulary space. 
Existing approaches often simply aggregate or interpolate feedback representations~\cite{li2025llm,li2023pseudo}, implicitly treating them as target points for query movement. 
In contrast, our DIFF models feedback vectors as confidence-weighted directional constraints rather than complete target representations. 
It also serves as a general reformulation mechanism on top of different feedback sources without retraining the underlying dense retriever.
\section{Conclusions}
In this paper, we introduce Directional Posterior Feedback Fusion (DIFF), a framework for feedback-based query reformulation in dense retrieval. 
Moving beyond conventional pointwise fusion, it models feedback as directional evidence in the tangent space of the original query, allowing each feedback item to constrain only the correction direction it supports while preserving complementary information from other feedback. 
Experiments across diverse retrieval settings demonstrate the effectiveness and robustness of DIFF and validate the benefits of directional modeling and confidence-aware fusion. 
Future work includes extending the directional formulation beyond local tangent-space modeling to capture more complex geometry in embedding spaces, and investigating richer mechanisms for estimating feedback reliability or jointly exploiting heterogeneous feedback sources.

\bibliographystyle{ACM-Reference-Format}
\bibliography{sample-base}

\appendix

\begin{algorithm}[t]\
\caption{Directional Posterior Feedback Fusion (DIFF)}\label{algo}
\begin{algorithmic}[1]
\REQUIRE Query vector $\mathbf q\in\mathbb R^D$; 
feedback vectors $\{\mathbf f_i\}_{i=1}^{K}$; 
prior strength $\lambda$; 
maximum update angle $\theta_{\mathrm{up}}$
\ENSURE Reformulated query $\mathbf q'$

\STATE \textbf{// Directional evidence~(\S\ref{subsec:direction})}
\FOR{$i=1,\ldots,K$}
    \STATE $s_i \leftarrow \mathbf q^\top\mathbf f_i$
    \STATE $\mathbf r_i \leftarrow \mathbf f_i-s_i\mathbf q$
    \STATE $\ell_i \leftarrow \|\mathbf r_i\|_2 ,\quad
           \mathbf u_i \leftarrow \mathbf r_i/\|\mathbf r_i\|_2$
\ENDFOR

\STATE \textbf{// Hierarchical confidence~(\S\ref{subsec:conf})}
\STATE $\{a_i\}_{i=1}^{K}
       \leftarrow \operatorname{RobustNorm}(\{s_i\}_{i=1}^{K})$
\STATE $\boldsymbol\mu \leftarrow \frac{1}{K}\sum_i\mathbf \ell_i\mathbf u_i$
\STATE $S \leftarrow \|\boldsymbol\mu\|_2^2,\quad
       V \leftarrow \frac{1}{K}\sum_i\|\ell_i\mathbf u_i-\boldsymbol\mu\|_2^2$
\STATE $g \leftarrow
       \dfrac{S}{S+V/K}
       \left(\frac{1}{K}\sum_i a_i\right)$
\STATE $\tau_i \leftarrow g a_i,\quad \forall i$

\STATE \textbf{// Posterior inference~(\S\ref{subsec:update})}
\STATE $U \leftarrow
       [\mathbf u_1,\ldots,\mathbf u_K]^\top,\quad
       W \leftarrow \operatorname{diag}(\sqrt{\tau_1},\ldots,\sqrt{\tau_K})$
\STATE $\boldsymbol\ell \leftarrow [\ell_1,\ldots,\ell_K]^\top$
\STATE Solve
       $\left(\lambda I_K+WUU^\top W\right)\mathbf y
       =W\boldsymbol\ell$
\STATE $\boldsymbol\delta^* \leftarrow U^\top W\mathbf y$

\STATE \textbf{// Spherical update~(\S\ref{subsec:update})}
\IF{$\|\boldsymbol\delta^*\|_2=0$}
    \STATE $\mathbf q'\leftarrow\mathbf q$
\ELSE
    \STATE $\phi \leftarrow \min\left(\|\boldsymbol\delta^*\|_2,\, g\theta_{\mathrm{up}}\right)$
    \STATE $\mathbf q' \leftarrow \cos(\phi)\mathbf q + \sin(\phi) \dfrac{\boldsymbol\delta^*}{\|\boldsymbol\delta^*\|_2}$
\ENDIF

\RETURN $\mathbf q'$
\end{algorithmic}
\end{algorithm}

\begin{figure}[htbp]
    \centering
    \begin{subfigure}[b]{0.48\textwidth}
        \centering
        \includegraphics[width=\textwidth]{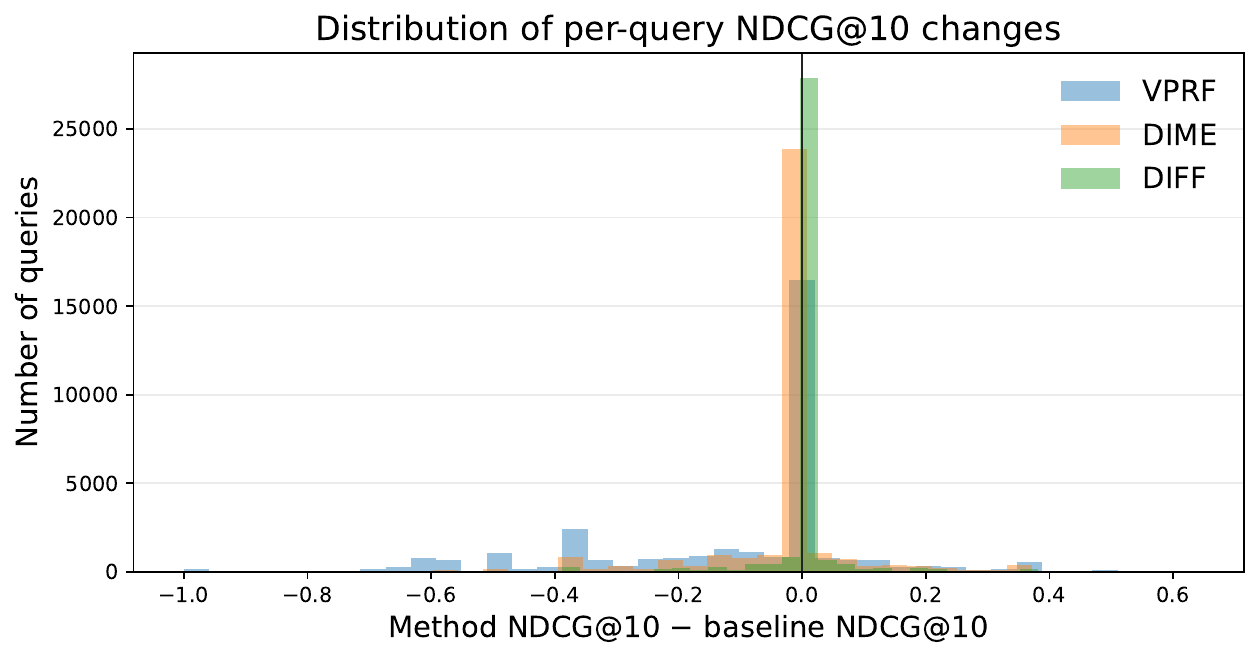}
        \caption{PRF $\Delta$NDCG@10 Distribution}
    \end{subfigure}
    \begin{subfigure}[b]{0.48\textwidth}
        \centering
        \includegraphics[width=\textwidth]{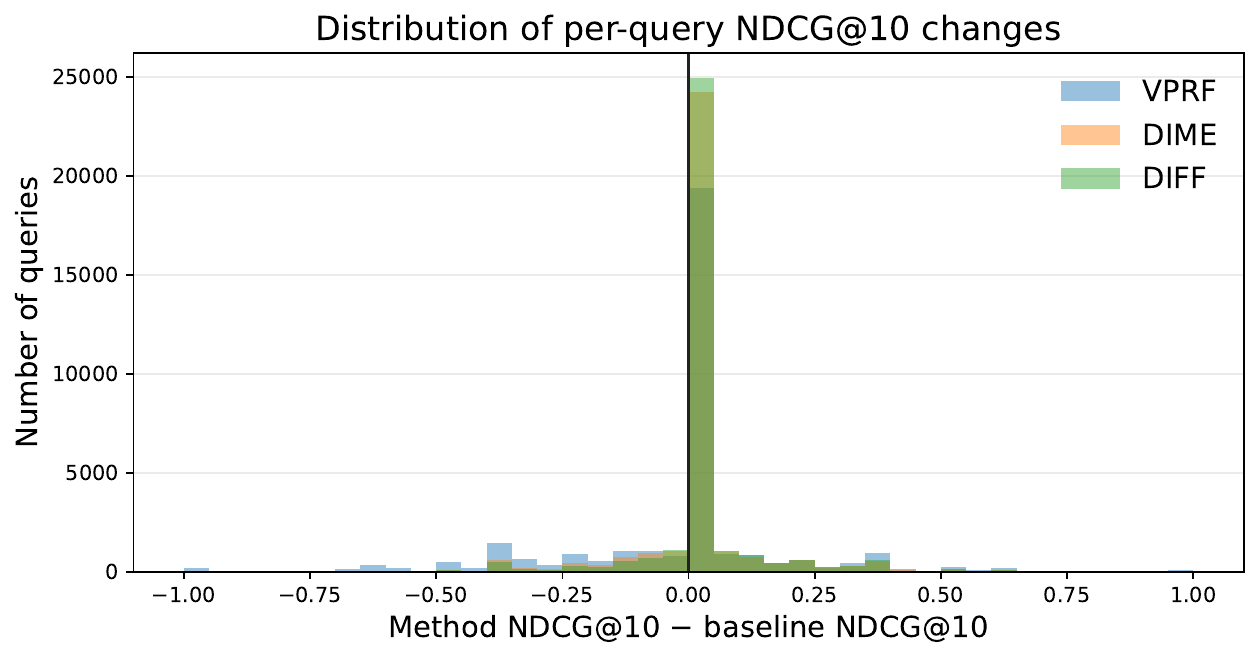}
        \caption{HyDE $\Delta$NDCG@10 Distribution}
    \end{subfigure}
    \caption{Distribution of NDCG@10 changes with BGE.}
    \label{fig:query_delta}
\end{figure}

\begin{figure}[htbp]
    \centering
     \begin{subfigure}[b]{0.48\textwidth}
        \centering
        \includegraphics[width=\textwidth]{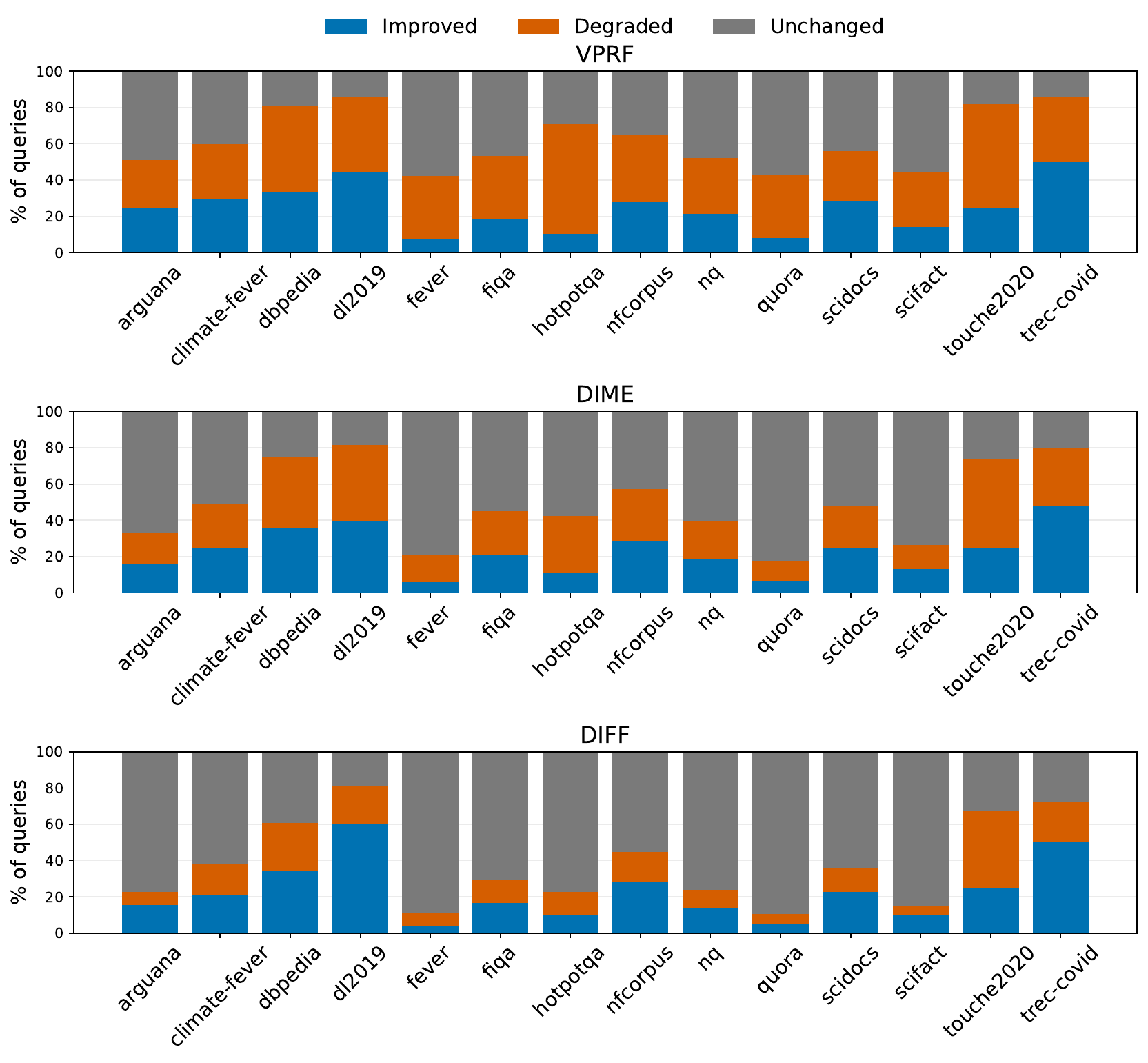}
        \caption{PRF $\Delta$NDCG@10}
    \end{subfigure}
    \begin{subfigure}[b]{0.48\textwidth}
        \centering
        \includegraphics[width=\textwidth]{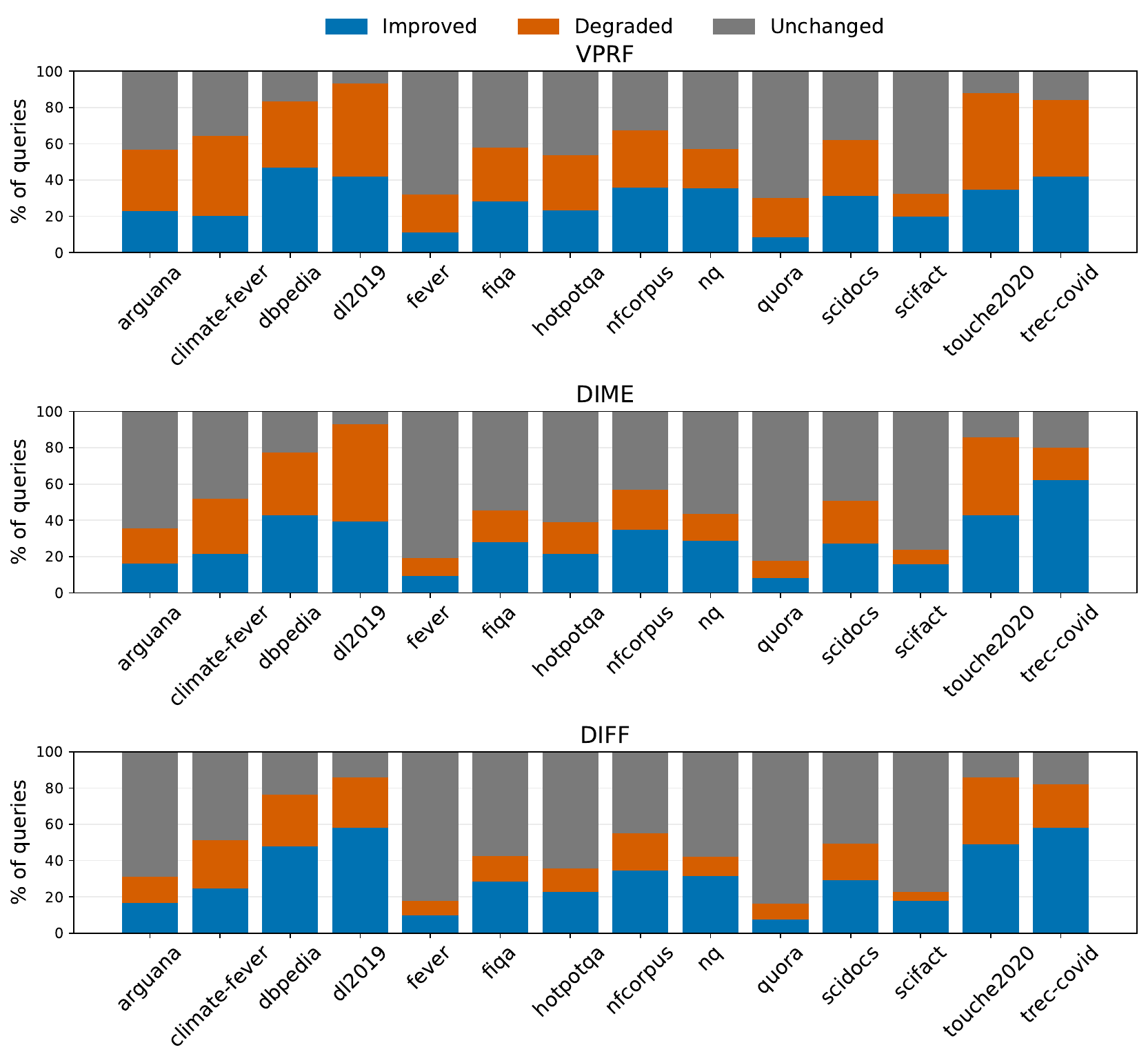}
        \caption{HyDE $\Delta$NDCG@10}
    \end{subfigure}
    \caption{NDCG@10 changes for all subsets with BGE.}
    \label{fig:subset_delta}
\end{figure}

\section{Mathematical Derivations} 
\subsection{Decomposition of Pointwise Feedback Fitting}\label{app:old_obj}
Let $\Pi_i=I-\mathbf u_i\mathbf u_i^\top$.
Since $\|\mathbf u_i\|=1$, $\Pi_i$ is the orthogonal projector onto $\mathbf u_i^{\perp}$: it is symmetric and idempotent, and $\mathbf u_i^\top\Pi_i=\mathbf 0^\top$.
Hence
\begin{equation}\label{eq:pf_split}
\begin{aligned}
\boldsymbol\delta-\ell_i\mathbf u_i&=(\mathbf u_i^\top\boldsymbol\delta-\ell_i)\,\mathbf u_i+\Pi_i\boldsymbol\delta ,\\
\|\boldsymbol\delta-\ell_i\mathbf u_i\|^2&=(\mathbf u_i^\top\boldsymbol\delta-\ell_i)^2+\|\Pi_i\boldsymbol\delta\|^2 ,
\end{aligned}
\end{equation}
where the second line follows from the first by the Pythagorean identity, because its two summands are orthogonal; this is Eq.~\eqref{eq:decomp}.
The first term measures the mismatch along the direction supported by feedback $i$; the second penalizes every component of $\boldsymbol\delta$ orthogonal to $\mathbf u_i$, including those supported by other items.
Fitting the full target $\ell_i\mathbf u_i$ therefore suppresses complementary corrections that item $i$ does not express, whereas DIFF keeps only the directional observation $\mathbf u_i^\top\boldsymbol\delta\approx\ell_i$.

\subsection{Proof of Theorem~\ref{thm:interference}}\label{app:proof}
Throughout, the assumptions of Theorem~\ref{thm:interference} are in force.
Write $A=\lambda I+\sum_i\tau_i\mathbf u_i\mathbf u_i^\top$, $P=\sum_i\tau_i\Pi_i$ and $c=\lambda+\sum_i\tau_i$, so that $A+P=cI$.
Items with $\tau_i=0$ enter none of $\mathcal L_{\mathrm{pt}}$, $\mathcal L_{\mathrm{dir}}$, $A$, $P$, $c$, $\mathbf b$; item $i$ is \emph{supported} if $\tau_i>0$ (at least one is, since $\mathbf b\neq\mathbf 0$), and the supported directions are \emph{collinear} if $\mathbf u_i\in\{\mathbf u,-\mathbf u\}$ for some unit vector $\mathbf u$ and all supported $i$.

\begin{lemma}\label{lem:structure}
(a) $\lambda I\preceq A\preceq cI$; in particular, $A$ is invertible.
(b) $\mathbf b\in T_{\mathbf q}$, and $A^{-1}$ maps $T_{\mathbf q}$ into itself.
(c) $P\mathbf v=\mathbf 0$ iff $\Pi_i\mathbf v=\mathbf 0$ for every supported $i$.
(d) $\mathcal L_{\mathrm{pt}}(\boldsymbol\delta)=\mathcal L_{\mathrm{dir}}(\boldsymbol\delta)+\tfrac12\boldsymbol\delta^\top P\boldsymbol\delta$ for all $\boldsymbol\delta$, and the minimizers over $T_{\mathbf q}$ are unique, coincide with the unconstrained minimizers, and satisfy:
\begin{equation}\label{eq:pf_minimizers}
\boldsymbol\delta^*_{\mathrm{dir}}=A^{-1}\mathbf b,\qquad
\boldsymbol\delta^*_{\mathrm{pt}}=\tfrac1c\,\mathbf b,\qquad
\boldsymbol\delta^*_{\mathrm{dir}}-\boldsymbol\delta^*_{\mathrm{pt}}=\tfrac1c\,A^{-1}P\,\mathbf b .
\end{equation}
\end{lemma}
\begin{proof}
(a) Since $0\preceq\mathbf u_i\mathbf u_i^\top\preceq I$ and $\tau_i\ge0$, we have $0\preceq\sum_i\tau_i\mathbf u_i\mathbf u_i^\top\preceq(\sum_i\tau_i)\,I$; adding $\lambda I$ gives the claim.
(b) $\mathbf b$ is a combination of the $\mathbf u_i\in T_{\mathbf q}$. As $\mathbf u_i^\top\mathbf q=0$, $A\mathbf q=\lambda\mathbf q$, so $A^{-1}\mathbf q=\lambda^{-1}\mathbf q$ and, for $\boldsymbol\delta\in T_{\mathbf q}$, $\mathbf q^\top A^{-1}\boldsymbol\delta=(A^{-1}\mathbf q)^\top\boldsymbol\delta=0$ by symmetry of $A^{-1}$.
(c) $\mathbf v^\top P\mathbf v=\sum_i\tau_i\|\Pi_i\mathbf v\|^2$ is a sum of non-negative terms: $P\mathbf v=\mathbf 0$ forces every term with $\tau_i>0$ to vanish, and conversely $\Pi_i\mathbf v=\mathbf 0$ for all supported $i$ gives $P\mathbf v=\sum_i\tau_i\Pi_i\mathbf v=\mathbf 0$.
(d) Applying Eq.~\eqref{eq:decomp} to each term of $\mathcal L_{\mathrm{pt}}$ in Eq.~\eqref{eq:two_obj} and using $\|\Pi_i\boldsymbol\delta\|^2=\boldsymbol\delta^\top\Pi_i\boldsymbol\delta$ gives the identity.
Differentiating, $\nabla\mathcal L_{\mathrm{dir}}=A\boldsymbol\delta-\mathbf b$ and $\nabla\mathcal L_{\mathrm{pt}}=(A+P)\boldsymbol\delta-\mathbf b=c\,\boldsymbol\delta-\mathbf b$, with Hessians $A\succ0$ and $cI\succ0$.
Both objectives are thus strictly convex with unique unconstrained minimizers $A^{-1}\mathbf b$ and $\mathbf b/c$; these lie in $T_{\mathbf q}$ by (b), so they are also the (unique) minimizers over $T_{\mathbf q}$.
Finally, $A^{-1}\mathbf b-\tfrac1c\mathbf b=\tfrac1c A^{-1}(cI-A)\mathbf b=\tfrac1c A^{-1}P\mathbf b$.
\end{proof}

\begin{proof}[Proof of Theorem~\ref{thm:interference}]
\emph{(i)} By Eq.~\eqref{eq:pf_minimizers}, the two minimizers coincide precisely when $P\mathbf b=\mathbf 0$.
Since $P$ is a sum of positive semidefinite matrices,
this condition is equivalent to
\begin{equation}\label{eq:pf_equiv}
\mathbf b=(\mathbf u_i^\top\mathbf b)\,\mathbf u_i,
\qquad \text{for every }i\text{ with }\tau_i>0.
\end{equation}
If the right-hand side holds, then $|\mathbf u_i^\top\mathbf b|=\|\mathbf b\|>0$ and $\mathbf u_i=\pm\,\mathbf b/\|\mathbf b\|$ for every supported $i$, i.e., the supported directions are collinear.
Conversely, if $\mathbf u_i=s_i\mathbf u$ with $s_i\in\{\pm1\}$ for all supported $i$, then $\mathbf b\in\operatorname{span}(\mathbf u)$ and $\Pi_i=I-\mathbf u\mathbf u^\top$ for all supported $i$, so $\Pi_i\mathbf b=\mathbf 0$ and Eq.~\eqref{eq:pf_equiv} gives $\boldsymbol\delta^*_{\mathrm{pt}}=\boldsymbol\delta^*_{\mathrm{dir}}$.

\emph{(ii)} Let the supported directions be mutually orthogonal and fix a supported $j$.
In $A\mathbf u_j=\lambda\mathbf u_j+\sum_i\tau_i\mathbf u_i(\mathbf u_i^\top\mathbf u_j)$ every term with $i\neq j$ vanishes (either $\tau_i=0$ or $\mathbf u_i^\top\mathbf u_j=0$), so $A\mathbf u_j=(\lambda+\tau_j)\mathbf u_j$ and $A^{-1}\mathbf u_j=\mathbf u_j/(\lambda+\tau_j)$.
Applying $A^{-1}$ to $\mathbf b=\sum_{i:\tau_i>0}\tau_i\ell_i\mathbf u_i$ term by term and using the orthonormality of the supported directions,
\begin{equation}\label{eq:pf_orth}
\begin{aligned}
\mathbf u_j^\top\boldsymbol\delta^*_{\mathrm{dir}}
&=\sum_{i:\tau_i>0}\frac{\tau_i\ell_i\,\mathbf u_j^\top\mathbf u_i}{\lambda+\tau_i}=\frac{\tau_j\ell_j}{\lambda+\tau_j},\\
\mathbf u_j^\top\boldsymbol\delta^*_{\mathrm{pt}}
&=\frac{\mathbf u_j^\top\mathbf b}{c}=\frac{\tau_j\ell_j}{\lambda+\tau_j+\sum_{i\neq j}\tau_i},
\end{aligned}
\end{equation}
which is Eq.~\eqref{eq:thm_interf}.
\end{proof}

\begin{corollary}\label{cor:norm}
Under the assumptions of Theorem~\ref{thm:interference}, $\|\boldsymbol\delta^*_{\mathrm{dir}}\|\ge\|\boldsymbol\delta^*_{\mathrm{pt}}\|$, with equality iff the supported directions are collinear.
\end{corollary}
\begin{proof}
Let $A=\sum_{k}\alpha_k\mathbf v_k\mathbf v_k^\top$ with orthonormal $\mathbf v_k$ and $\alpha_k\in[\lambda,c]$ (Lemma~\ref{lem:structure}(a)).
Then
\begin{equation*}
\|\boldsymbol\delta^*_{\mathrm{dir}}\|^2=\|A^{-1}\mathbf b\|^2=\sum_k\alpha_k^{-2}(\mathbf v_k^\top\mathbf b)^2\ \ge\ c^{-2}\|\mathbf b\|^2=\|\boldsymbol\delta^*_{\mathrm{pt}}\|^2 ,
\end{equation*}
with equality iff $\mathbf v_k^\top\mathbf b=0$ whenever $\alpha_k<c$, i.e., iff $A\mathbf b=c\mathbf b$, i.e., iff $P\mathbf b=\mathbf 0$, which is collinearity by Eq.~\eqref{eq:pf_equiv}.
\end{proof}

\section{Efficient Solution to the DIFF Objective} \label{app:sol}
DIFF estimates the query correction by solving the objective in Eq.~\eqref{eq:obj}.
Let $U= [\mathbf u_1^\top, \mathbf u_2^\top, ...,  \mathbf u_K^\top]^\top \in\mathbb R^{K\times D}, T=\operatorname{diag}(\tau_1,\ldots,\tau_K), \boldsymbol\ell=[\ell_1,\ldots,\ell_K]^\top$, the objective can be written in matrix form as:
\begin{equation}
    \mathcal L(\boldsymbol\delta)=\frac{\lambda}{2}\|\boldsymbol\delta\|_2^2+\frac12 (U\boldsymbol\delta-\boldsymbol\ell)^\top T(U\boldsymbol\delta-\boldsymbol\ell).
\end{equation}
Taking the gradient with respect to $\boldsymbol\delta$ and setting the gradient to zero yields:
\begin{equation}
    \nabla_{\boldsymbol\delta}\mathcal L = \lambda\boldsymbol\delta + U^\top T (U\boldsymbol\delta-\boldsymbol\ell).
\end{equation}
\begin{equation}
    \lambda\boldsymbol\delta+U^\top TU\boldsymbol\delta=U^\top T\boldsymbol\ell.
\end{equation}

A direct solution is therefore $\boldsymbol\delta^*=(\lambda I_D+U^\top TU)^{-1}U^\top T\boldsymbol\ell$. 
However, it involves solving a $D\times D$ linear system. 
Since dense embedding dimensions can be large while the number of feedback items typically satisfies $K\ll D$, it is preferable to reformulate the solution in terms of a $K\times K$ system.
Let $W=T^{1/2} = \operatorname{diag} (\sqrt{\tau_1},\ldots,\sqrt{\tau_K})$ and $A=WU$, we have $U^\top TU=U^\top W^2U=A^\top A, U^\top T\boldsymbol\ell = A^\top W\boldsymbol\ell$, thus $\boldsymbol\delta^*=(\lambda I_D+A^\top A)^{-1}A^\top W\boldsymbol\ell$.
With $(\lambda I_D+A^\top A)^{-1}A^\top=A^\top(\lambda I_K+AA^\top)^{-1}$, we can obtain Eq.~\eqref{eq:opt_result}:
\begin{equation}
    \boldsymbol\delta^*=A^\top(\lambda I_K+AA^\top)^{-1}W\boldsymbol\ell=U^\top W\left(\lambda I_K+WUU^\top W\right)^{-1}W\boldsymbol\ell.
\end{equation}

\begin{table*}[htbp]
\centering
\caption{NDCG@10 results on each subset with BGE.}
\label{tab:fullresult}
\resizebox{0.9\textwidth}{!}{
\begin{tabular}{cccccccccccccccc}
\toprule
\multicolumn{1}{c|}{Subset} & AA & CF & DB & FE & FQ & HP & MS & NF & NQ & QU & SD & SF & TO & \multicolumn{1}{c|}{TR} & Avg. \\
\midrule
\multicolumn{1}{c|}{BGE} & 0.4534 & 0.3119 & 0.4078 & 0.8630 & 0.4064 & 0.7260 & 0.7245 & 0.3732 & 0.5411 & 0.8874 & 0.2168 & 0.7394 & 0.2571 & \multicolumn{1}{c|}{0.7802} & 0.5492 \\
\midrule
\multicolumn{16}{l}{\textit{Feedback from PRF}} \\
\midrule
\multicolumn{1}{c|}{VPRF} & 0.4552 & 0.3008 & 0.3831 & 0.7423 & 0.3584 & 0.5834 & 0.7268 & 0.3479 & 0.4976 & 0.7654 & 0.2164 & 0.6740 & 0.2132 & \multicolumn{1}{c|}{0.7983} & 0.5045 \\
\multicolumn{1}{c|}{Concat-P}  & 0.4549 & 0.3032 & 0.3852 & 0.7859 & 0.3791 & 0.6616 & 0.7275 & 0.3517 & 0.5294 & 0.8611 & 0.2160 & 0.7044 & 0.2417 & \multicolumn{1}{c|}{0.8096} & 0.5294 \\
\multicolumn{1}{c|}{Concat-L} & 0.4666 & 0.2757 & 0.3701 & 0.8376 & 0.3688 & 0.6879 & 0.7414 & 0.3445 & 0.5342 & 0.7747 & 0.2123 & 0.7178 & 0.2549 & \multicolumn{1}{c|}{0.7607} & 0.5248 \\
\multicolumn{1}{c|}{OR-PRF} & 0.4551 & 0.3077 & 0.3954 & 0.8153 & 0.3890 & 0.6824 & 0.7356 & 0.3664 & 0.5380 & 0.8650 & 0.2197 & 0.7188 & 0.2488 & \multicolumn{1}{c|}{0.8116} & 0.5392 \\
\multicolumn{1}{c|}{DIME} & 0.4509 & 0.3082 & 0.4005 & 0.8347 & 0.4018 & 0.6857 & 0.7244 & 0.3684 & 0.5287 & 0.8717 & 0.2186 & 0.7347 & 0.2323 & \multicolumn{1}{c|}{0.8049} & 0.5404 \\
\multicolumn{1}{c|}{Eclipse} & 0.4478 & 0.3084 & 0.3984 & 0.8423 & 0.4037 & 0.6919 & 0.7068 & 0.3679 & 0.5246 & 0.8764 & 0.2119 & 0.7314 & 0.2350 & \multicolumn{1}{c|}{0.7967} & 0.5388 \\
\multicolumn{1}{c|}{AnchorQE} & 0.4558 & 0.3171 & 0.4065 & 0.8477 & 0.4082 & 0.7015 & 0.7311 & 0.3697 & 0.5398 & 0.8780 & 0.2223 & 0.7382 & 0.2363 & \multicolumn{1}{c|}{0.7999} & 0.5466 \\
\multicolumn{1}{c|}{DIFF} & 0.4673 & 0.3140 & 0.4104 & 0.8541 & 0.4076 & 0.7212 & 0.7343 & 0.3772 & 0.5451 & 0.8834 & 0.2223 & 0.7455 & 0.2473 & \multicolumn{1}{c|}{0.8012} & 0.5522 \\ 
\midrule
\multicolumn{16}{l}{\textit{Feedback from HyDE}} \\
\midrule
\multicolumn{1}{c|}{VPRF} & 0.4298 & 0.2385 & 0.4320 & 0.8249 & 0.4073 & 0.7058 & 0.7396 & 0.3865 & 0.5937 & 0.8302 & 0.2135 & 0.7708 & 0.2350 & \multicolumn{1}{c|}{0.7664} & 0.5410 \\
\multicolumn{1}{c|}{Concat-P} & 0.4516 & 0.2538 & 0.4263 & 0.8102 & 0.4111 & 0.7133 & 0.7423 & 0.3844 & 0.5938 & 0.8645 & 0.2191 & 0.7734 & 0.2639 & \multicolumn{1}{c|}{0.7836} & 0.5494 \\
\multicolumn{1}{c|}{Concat-L} & 0.4488 & 0.2415 & 0.4084 & 0.8007 & 0.3909 & 0.6710 & 0.7189 & 0.3649 & 0.5746 & 0.8284 & 0.2097 & 0.7497 & 0.2594 & \multicolumn{1}{c|}{0.7562} & 0.5302  \\
\multicolumn{1}{c|}{OR-PRF} & 0.4527 & 0.2687 & 0.4382 & 0.8306 & 0.4248 & 0.7303 & 0.7567 & 0.3943 & 0.6014 & 0.8685 & 0.2224 & 0.7713 & 0.2679 & \multicolumn{1}{c|}{0.8203} & 0.5606 \\
\multicolumn{1}{c|}{DIME} & 0.4483 & 0.2990 & 0.4310 & 0.8641 & 0.4285 & 0.7366 & 0.7344 & 0.3899 & 0.5805 & 0.8829 & 0.2201 & 0.7652 & 0.2547 & \multicolumn{1}{c|}{0.8407} & 0.5626 \\
\multicolumn{1}{c|}{Eclipse} & 0.4472 & 0.2982 & 0.4335 & 0.8634 & 0.4327 & 0.7378 & 0.7372 & 0.3895 & 0.5819 & 0.8829 & 0.2198 & 0.7690 & 0.2534 & \multicolumn{1}{c|}{0.8508} & 0.5641  \\
\multicolumn{1}{c|}{AnchorQE} & 0.4592 & 0.3137 & 0.4424 & 0.8723 & 0.4326 & 0.7483 & 0.7536 & 0.3885 & 0.5872 & 0.8870 & 0.2259 & 0.7699 & 0.2580 & \multicolumn{1}{c|}{0.8401} & 0.5700 \\
\multicolumn{1}{c|}{DIFF} & 0.4550 & 0.3119 & 0.4429 & 0.8695 & 0.4342 & 0.7487 & 0.7592 & 0.3903 & 0.5962 & 0.8832 & 0.2277 & 0.7738 & 0.2646 & \multicolumn{1}{c|}{0.8480} & 0.5717 \\
\bottomrule
\end{tabular}
}
\end{table*}
\begin{table}[t]
\centering
\caption{Hyperparameter analysis of the prior strength $\lambda$ and maximum update angle $\theta_{\mathrm{up}}$ with BGE. Results are reported in NDCG@10.}
\label{tab:hyper}
\resizebox{\columnwidth}{!}{
\begin{tabular}{c|ccc|ccc}
\toprule
& \multicolumn{3}{c|}{PRF} & \multicolumn{3}{c}{HyDE} \\
\cmidrule(lr){2-4}\cmidrule(lr){5-7}
$\theta_{\mathrm{up}} \backslash \lambda$
& 0.5 & 1 & 2
& 0.5 & 1 & 2 \\
\midrule
$\pi/12$ & 0.5518 & \textbf{0.5522} & 0.5517
         & 0.5694 & 0.5667 & 0.5631 \\
$\pi/6$  & 0.5452 & 0.5491 & 0.5512
         & 0.5682 & 0.5711 & 0.5689 \\
$\pi/4$  & 0.5453 & 0.5462 & 0.5496
         & 0.5665 & 0.5695 & \textbf{0.5717} \\
\bottomrule
\end{tabular}
}
\end{table}
\begin{table}[t]
\centering
\caption{Efficiency comparison of different query reformulation methods on 10 queries.}
\label{tab:time}
\begin{tabular}{lcc}
\toprule
Method & Time (s) & Relative Time \\
\midrule
VPRF & 0.247 & 0.96$\times$ \\
Concat-P & 0.265 & 1.03$\times$ \\
DIME & 0.239 & 0.93$\times$ \\
AnchorQE & 0.946 & 3.67$\times$ \\
DIFF & 0.258 & 1.00$\times$ \\
\bottomrule
\end{tabular}
\end{table}

This form requires solving only a $K\times K$ system instead of a $D\times D$ one. Constructing $UU^\top$ costs $\mathcal{O}(K^2D)$, and solving the resulting linear system costs $\mathcal{O}(K^3)$, leading to an overall complexity of $\mathcal{O}(K^2D+K^3)$.
Since $K\ll D$ in typical retrieval settings, this formulation is substantially more efficient than solving directly in the embedding dimension.

\section{DIFF Algorithm} \label{app:algo}
Algorithm~\ref{algo} summarizes the overall DIFF workflow. Generally, it consists of 4 steps: calculating the directional evidence, calculating the hierarchical confidence, posterior inference (solving the optimization objective), and spherical update.

\section{Implementation Details} \label{app:implement}
For dense retrieval, we use a batch size of 256 and perform inference in FP16 precision.
For all query reformulation methods, hyperparameters are selected via grid search to ensure a fair comparison. 
Specifically for DIFF, we search the prior strength over $\lambda \in \{0.5, 1, 2\}$ and the maximum update angle over $\theta_{\mathrm{up}} \in \{\pi/12, \pi/6, \pi/4\}$. 
The vector indexing is implemented using the Faiss-GPU framework~\cite{douze2025faiss}.
All experiments are conducted on a single NVIDIA A800 GPU and averaged over 5 runs.

\section{Extensive Experimental Results} 
\subsection{NDCG@10 on Subsets} \label{app:subset_detail}
Table~\ref{tab:fullresult} reports the detailed NDCG@10 on all subsets. 

\subsection{NDCG@10 Changes ($\Delta$NDCG@10)} \label{app:change}
Figure~\ref{fig:query_delta} reports the overall $\Delta$NDCG@10 distribution, and Figure~\ref{fig:subset_delta} demonstrates the $\Delta$NDCG@10 for each subset.
Under PRF, VPRF degrades substantially more queries than it improves, reflecting the risk of directly incorporating noisy feedback. 
In contrast, DIFF leaves most queries unchanged and keeps the numbers of improved and degraded queries nearly balanced, with its $\Delta$NDCG@10 distribution tightly concentrated around zero. 
Under HyDE, where the feedback quality is generally higher, DIFF produces more improvements than degradations while still preserving a large fraction of queries. 
It suggests that DIFF behaves conservatively under noisy feedback, but allows more beneficial updates when the feedback is more reliable.

\subsection{Hyperparameter Analysis ($\lambda$ and $\theta_{\mathrm{up}}$)} \label{app:hyper}
We analyze the sensitivity of DIFF to the prior strength $\lambda$ and the maximum update angle $\theta_{\mathrm{up}}$, as shown in Table~\ref{tab:hyper}. 
The two feedback sources exhibit different preferences. 
Under PRF, a relatively small update bound ($\theta_{\mathrm{up}}=\pi/12$) performs consistently well across different $\lambda$, indicating that conservative updates are preferable when PRF may contain substantial noise. 
In contrast, HyDE benefits from allowing larger query movements, with competitive performance observed under $\theta_{\mathrm{up}}=\pi/6$ and $\pi/4$, consistent with its generally higher-quality feedback. 
Meanwhile, increasing $\lambda$ tends to mitigate the effect of a larger update bound by imposing a stronger prior toward the original query.
Overall, the interaction between $\lambda$ and $\theta_{\mathrm{up}}$ reflects their complementary roles in controlling reformulation: the former regularizes the inferred correction, while the latter explicitly bounds the final angular displacement.

\subsection{Efficiency Analysis} \label{app:efficiency}
We compare the efficiency of different methods by measuring their execution time on 10 queries.
As shown in Table~\ref{tab:time}, DIFF achieves comparable efficiency to lightweight baselines such as VPRF, Concat-P, and DIME.
In contrast, AnchorQE takes approximately 3.67 times that of DIFF, due to its additional retrieval operations for estimating feedback contributions.
These results demonstrate that DIFF effectively incorporates directional evidence without introducing substantial computational overhead.
\end{document}